\documentclass[onecolumn,journal]{IEEEtran}
\ifCLASSINFOpdf
\else
\fi

\usepackage[T1]{fontenc}
\usepackage[latin9]{inputenc}
\usepackage{color}
\usepackage{array}
\usepackage{float}
\usepackage{mathrsfs}
\usepackage{mathtools}
\usepackage{amsthm}
\usepackage{amstext}
\usepackage{amssymb}
\usepackage{stmaryrd}
\usepackage{graphicx}
\usepackage{pgfplots}
\usepackage{mathcomp}

\makeatletter

\theoremstyle{plain}
\newtheorem{thm}{\protect\theoremname}
\theoremstyle{plain}
\newtheorem{prop}[thm]{\protect\propositionname}
\theoremstyle{plain}
\newtheorem{cor}[thm]{\protect\corollaryname}
\theoremstyle{plain}
\newtheorem{lem}[thm]{\protect\lemmaname}
\theoremstyle{definition}
\newtheorem{example}[thm]{\protect\examplename}
\theoremstyle{definition}
\newtheorem{defn}[thm]{\protect\definitionname}

\AtBeginDocument{
	
}

\makeatother

  \providecommand{\corollaryname}{Corollary}
  \providecommand{\examplename}{Example}
  \providecommand{\lemmaname}{Lemma}
  \providecommand{\propositionname}{Proposition}
  \providecommand{\theoremname}{Theorem}
  \providecommand{\definitionname}{Definition}

\newcommand{\Div}{\operatorname{div}}
\newcommand{\Diff}{\operatorname{Diff}}

\newcommand{\Con}{\operatorname{Con}}

\newcommand{\supp}{\operatorname{supp}}
\newcommand{\Tr}{\operatorname{Tr}}
\newcommand{\ceil}[1]{\left\lceil {#1} \right\rceil}
\newcommand{\floor}[1]{\left\lfloor {#1} \right\rfloor}
\begin{document}

%
% paper title
% Titles are generally capitalized except for words such as a, an, and, as,
% at, but, by, for, in, nor, of, on, or, the, to and up, which are usually
% not capitalized unless they are the first or last word of the title.
% Linebreaks \\ can be used within to get better formatting as desired.
% Do not put math or special symbols in the title.
\title{Explicit Construction of AG Codes from Generalized Hermitian Curves}
%
%
% author names and IEEE memberships
% note positions of commas and nonbreaking spaces ( ~ ) LaTeX will not break
% a structure at a ~ so this keeps an author's name from being broken across
% two lines.
% use \thanks{} to gain access to the first footnote area
% a separate \thanks must be used for each paragraph as LaTeX2e's \thanks
% was not built to handle multiple paragraphs
%

\author{Chuangqiang~Hu    % <-this % stops a space
	
\thanks{C. Hu is with the School
	of Mathematics and Computational Science,  Sun Yat-sen University, Guangzhou 510275, P.R.China.\protect\\
	\protect\\
	% note need leading \protect in front of \\ to get a newline within \thanks as
	% \\ is fragile and will error, could use \hfil\break instead.
	E-mail: huchq@mail2.sysu.edu.cn}% <-this % stops a % <-this % stops a space
\thanks{Manuscript received *********; revised ********.}% <-this % stops a space
} 

\maketitle

% As a general rule, do not put math, special symbols or citations
% in the abstract or keywords.
\begin{abstract}
We present   multi-point algebraic geometric codes overstepping the Gilbert-Varshamov   bound. The construction is based on the  generalized Hermitian curve introduced by A.~Bassa, P.~Beelen, A.~Garcia, and H.~Stichtenoth. These codes are described in detail by constrcting a generator matrix. It 
turns out that these codes have nice properties similar to those of Hermitian codes.  It is shown that the duals are also such codes and an explicit formula is given.    
\end{abstract}

% Note that keywords are not normally used for peerreview papers.

% For peer review papers, you can put extra information on the cover
% page as needed:
% \ifCLASSOPTIONpeerreview
% \begin{center} \bfseries EDICS Category: 3-BBND \end{center}
% \fi
%
% For peerreview papers, this IEEEtran command inserts a page break and
% creates the second title. It will be ignored for other modes.
\IEEEpeerreviewmaketitle

\section{Introduction}
% The very first letter is a 2 line initial drop letter followed
% by the rest of the first word in caps.
% 
% form to use if the first word consists of a single letter:
% \IEEEPARstart{A}{demo} file is ....
% 
% form to use if you need the single drop letter followed by
% normal text (unknown if ever used by IEEE):
% \IEEEPARstart{A}{}demo file is ....
% 
% Some journals put the first two words in caps:
% \IEEEPARstart{T}{his demo} file is ....
% 
% Here we have the typical use of a "T" for an initial drop letter
% and "HIS" in caps to complete the first word.
\IEEEPARstart{L}{et}   $  F   $ be a function field over a finite field $ \mathbb{F}_l  $. An algebraic geometric code is of the form $ C_{\mathscr{L}}(  D, G) $ with $D=P_{1}+\ldots+P_{n}$ where the $ P _ j $'s are pairwise-distinct places of degree one in $  F $, and $G$ is a divisor of  $  F $ such that $ \supp (G ) \cap \supp (D ) = \emptyset $. And  $ C_{\mathscr{L}}(  D, G) $ is defined by
\[C_{\mathscr{L}}(  D, G):={\left\{ (f(P_{1}), f(P_{2}),\ldots, f(P_{n})) \left| f\in\mathscr{L}(G) \right. \right\} },\]
where $\mathscr{L}(G)$ denotes the Riemann-Roch space associated to $ G $, see \cite{Stichtenoth,Niederreiter} as general references for all facts concerning algebraic geometric (AG) codes.

The Gilbert-Varshamov (GV) bound  \cite{Stichtenoth,Niederreiter}  guarantees the existence of families of codes over the finite field  with good asymptotic parameters; i.e., information rate and relative minimum distance. It is well known that   the parameters of   AG codes related to   asymptotically good towers of function fields are better than the GV bound in a certain range of the rate \cite{Tsfasman1982,Niederreiter2000}.

 Denote by $ N(F) $ the number of rational places of $ F/\mathbb{F}_l $.  Let $ N_l(g):= \max \{ N(F) | \text{$ F $ is a function field over $ \mathbb{F}_l  $ of genus $ g $}\} $. The real number
 \[ A(l):= \limsup_{g\to \infty} \frac{N_l(g)}{g} ,\]
 is called Ihara's quantity. The Drinfeld-Vladut bound \cite{Drinfeld} tells  us that  
 \begin{equation} \label{eq:Ihara}
 A(l)\leqslant \sqrt l-1 .
 \end{equation}
 If $ l $ is a square; then
\begin{equation}\label{eq:square}
   A(l)=\sqrt l-1 ,
\end{equation}
  which was first shown by Ihara \cite{Ihara}. Tsfasman, Vladut and Zink gave in \cite{Tsfasman} an independent proof
  of Equation (\ref{eq:square}). 
  For $ l  $ is a square, and $ l \geqslant 49$, the GV bound was improved by the famous  Tsfasman-Vladut-Zink theorem \cite{Tsfasman}.
    Also, for $ l  = q^c $ with odd $ c > 1  $ and very large $ q $, there are improvements
   of the GV bound due to Niederreiter and Xing \cite{Niederreiter}.

  For
  applications to coding theory though, explicit construction of good towers  are needed. In \cite{Garcia}  
  A.~Garcia, and H.~Stichtenoth gave an explicit  construction of a tower of Artin-Schreier extensions of function fields over $ \mathbb{F}_{q^2} $ attaining the Drinfeld-Vladut bound.
Recently,  in \cite{Bassa}, A.~Bassa, P.~Beelen, A.~Garcia, and H.~Stichtenoth    produced an explicit tower of
function fields over finite fields $ \mathbb{F}_{q^{2b+1}} $ for any   integer $ b \geqslant 1 $ and
showed that this tower gives
\begin{equation}\label{eq:Ihara2}
 A(q^{2b+1})  \geqslant \frac{2 (q^{b+1}-1)}{q + 1 + \epsilon}
\quad  \text{with }
\epsilon =  \frac{q-1}{q^b -1 }.
\end{equation}
Using this tower they obtained an improvement of  the GV bound 
 for all non-prime fields $ \mathbb{F}_l $ with $ l \geqslant 49 $, except possibly $ l= 125 $ in \cite{Bassa2014}.

Their construction can be restated as follows. 
Let $ \mathbb{F}_l $ be a non-prime field and write $ l =q^c $ with $ c \geqslant 2 $. Here the integer $  c $ can be even or odd. For every partition of $ c $ in relatively prime parts; i.e., 
 \[ c= a+b \quad \text{ with $a\geqslant 1, b\geqslant 1 $} , \]
   where the greatest common
    factor of $ a $ and $ b $ is $\gcd (a,b)=1$, we define 
 \begin{align*}
 H(x,y)& :=\frac{y^{q^{a}}}{x}+\frac{y^{q^{a+1}}}{x^{q}}+\ldots+\frac{y^{c-1}}{x^{q^{b-1}}}   \\ 
 &\quad +\frac{y }{x^{q^{b}}}+\frac{y^{q}}{x^{q^{b+1}}}+\ldots+\frac{y^{q^{a-1}}}{x^{q^{c-1}}}  .
 \end{align*}
The asymptotically good  tower $ \mathcal{F} =(F^{(1)} \subseteq F^{(2)} \subseteq F^{(3)} \subseteq \ldots ) $ over $ \mathbb{F}_l $   is recursively given by the equation
\begin{equation}\label{eq:the_plane_curve}
 H(x,y)=1 .
\end{equation}
Precisely speaking,
\begin{enumerate}
\item $ F^{(1)} = \mathbb{F}_l (x_1)  $ is the rational function field, and
\item $ F^{(i+1)} =  F^{(i )} (x_{i+1}) $ with $ H(x_i, x_{i+1})=1 $, for all $i \geqslant 1$.
\end{enumerate}
For the case $  a=b+1 $, this tower implies Equation (\ref{eq:Ihara2}). For $a = b =1 $, this tower is identical with the one constructed in \cite{Garcia} and the second function field $ F^{(2)} $ is the Hermitian function field \cite{Stichtenoth,Tiersma}. To see this, we restate Equation (\ref{eq:the_plane_curve}) as follows
\[ \frac{y^{q }}{x} + \frac{y }{x^{q }} =1.
 \]
We replace $ xy $ by $ z $; then
\begin{equation}
z+z^q=x^{q+1},
\end{equation}
which is exactly the canonical definition of Hermitian curve.
So  the second function field $ F^{(2)} $ with general coefficients $a, b$ can be regarded as the 
generalized Hermitian function field. And then   the generalized Hermitian curve can be defined as follows.
\begin{defn}
   The \textbf{generalized Hermitian curve} over $ \mathbb{F}_{q^c}  $ with coprime integers $a, b\in \mathbb{Z}^{+}$ verifying $a+b=c$, is  defined by the affine equation
    \begin{align*}
    \frac{y^{q^{a}}}{x}+\frac{y^{q^{a+1}}}{x^{q}}+\ldots+\frac{y^{c-1}}{x^{q^{b-1}}}   +\frac{y }{x^{q^{b}}}+\frac{y^{q}}{x^{q^{b+1}}}+\ldots+\frac{y^{q^{a-1}}}{x^{q^{c-1}}} =1 .
    \end{align*}
\end{defn}

 The AG codes arising from the Hermitian curve are widely investigated, which are called Hermitian codes. The advantage
 of these codes is that these codes are easy to describe and to encode and decode. Moreover, these codes often have excellent parameters.
 
  One-point codes from Hermitian curves were well-studied in the literature, and efficient methods to decode them were known  \cite{Stichtenoth,Guruswami,Yang,Yang2}. The minimum distance of Hermitian two-point codes had been first determined by M.~Homma and S.~J.~Kim  \cite{Homma,Homma2,Homma3,Homma4}. The explicit formulas for the dual minimum distance of such codes were given by S.~Park  \cite{Park}. Recently, Hermitian codes from higher-degree places had been considered in \cite{Korchmaros}. The dual minimum distance of many three-point codes from Hermitian curves was computed in \cite{Ballico}, by extending a recent and powerful approach by A.~Couvreur \cite{Couvreur}.  H.~Maharaj, G.~L.~Matthews and G.~Pirsic determined explicit bases for large classes of Riemann-Roch spaces of the Hermitian function field \cite{Maharaj}. These bases gave better estimates on the parameters of a large class of multi-point Hermitian codes. 
 
In \cite{Hu}, the authors  explicitly constructed multi-point codes from the generalized Hermitian curves with coefficients $a=1$, and $b=c-1$.

In this paper we investigate multi-point codes from the   generalized Hermitian curves $ \mathcal{X}$ with $ a = b+1 $. The advantage in this setting is that the related tower  achieve the bound stated in Equation (\ref{eq:Ihara2}) as indicated in \cite{Bassa}.
 We introduce four important divisors as follows according to  \cite{Bassa},
\begin{enumerate}
		\item $D: = \sum_{\alpha, \beta } {D_{\alpha, \beta }} $, where $D_{\alpha,\beta}:=(x=\alpha, y=\beta)$ with $ \alpha , 
		\beta \in \mathbb{F}_ {q^c}^{*}$ satisfying $ H(\alpha, \beta)=1 $;
		\item $P:=(x=0,y=0)$, which would be split into two parts, namely $ P= P_1 + P_0 $ where $P_1: =(x=0, y=0, x^{-q^b} y = a^{-1})$ denotes a rational place, if $ a $ and $q $ are coprime;
		\item $Q:=(x=\infty, y= \infty )$; 
		\item $V:=(x=0, y=\infty)$.
\end{enumerate}
  The divisors $ D $, $ P $, $ Q $ and $ V $ contain all the possible rational places on the curve. We define the algebraic geometric codes over   $ \mathbb{F}_{q^c}  $
	 \[
	C_{v,r,s,t}=C_{\mathscr{L}}(  D, v P_1 + rP_0 + sQ +tV ) .  \] 

For applications of such codes in practice one needs an explicit description, which means an explicit basis for the vector space   $ \mathscr L(vP_1 + rP_0 +s Q+t V) $ or a generator matrix of the code $ C_{v,r,s,t} $. We discover that this problem is related to a point-counting problem. Pick's theorem \cite{Haigh,Varberg} provides a simple formula for calculating for two-dimensional lattice point set. Let $ \Omega $ be a lattice polygon. Assume there are $ I $ lattice points in the interior of $ \Omega $, and $ M $ lattice points on its boundary. Let $ S $ denote the area of $ \Omega $. Then
    \[ 
    S = I + \frac{M}{2} - 1. 
    \] 
	 The main technical part in this paper is to solve the related three-dimensional point-counting problem using Pick's theorem. So we can describe the code $ C_{v,r,s,t} $ by constructing a generaor matrix. As in the Hermitian case, it turns out that the dual code of   $ C_{v,r,s,t} $  is of the same type. Finally, it is shown that  the Goppa bound of $ C_{v,r,s,t} $ improves the GV bound in a certain interval.
	 For example, we find a $ [496,250,\geqslant 172]$-code over $ \mathbb{F}_{32} $ overstepping the GV bound.
	 
	 The paper is organized as follows. In Section 2, we introduce some arithmetic properties of the curve $ \mathcal{X} $ and describe all the rational places. In Section 3, we construct a basis for the Riemann-Roch space $ \mathscr{L}(vP_1 + r P_0 + sQ +tV ) $. Section 4 is devoted to investigating the parameters and the duality properties of $ C_{v,r,s,t} $.

\section{the arithmetic properties of the curve}
 We start with some notations according to \cite{Bassa}. Let $ q  $ be a power of a prime $p $ and $ \mathbb{F}_{q^{c}} $ be a finite field of cardinality $ q^{c}  $. For an integer $ a \geqslant 1 $, we define the   function 
\[ \Tr_a(x):= x+ x^{q^1}+ x^{q^2} +\ldots+ x^{q^{a-1}}.
 \]
 We assume $c$ is an odd number and fix a partition of $ c $ into two consecutive integers; i.e., we write
 \[
 c=a+b, \text{ with  $ a=b+1 , b \in \mathbb{Z}^+$}. 
  \]
  
In this section we study the generalized Hermitian curve $ \mathcal{X} $ over $\mathbb{F}_{q^{c}} $
\begin{equation}\label{eq:curve_x_y}
\Tr_{b}(\frac{y^{q^ a}}{x})+\Tr_{a}(\frac{y}{x^{q^b}})=1.
\end{equation} 
For abbreviation we set $N_k:=( q^k-1)/(q-1)$ for every integer $ k>1 $.  For an element $f$ in the function field $\mathbb{F}_{q^c } (\mathcal X) $ of $\mathcal X $, define 
\[ 
\Div(f), \Div_0(f),\text{ and}   \Div_\infty(f)  
\]
the principal divisor, zero divisor and the pole divisor of $f$ in $\mathbb{F}_{q^c } (\mathcal X) $. Let $P:=(x=0,y=0)$, $Q:=(x=\infty,y=\infty)$, and $V:=(x=0,y=\infty)$
be the divisors of   $\mathbb{F}_{q^{c}}(\mathcal{X})$. For a divisor $D$ in   $\mathbb{F}_{q^c }$, we denote by $\deg (D)$ the degree of $D$. Several results in \cite{Bassa} are restated by the following proposition.
\begin{prop}[\cite{Bassa}]
	\label{prop:valuation}
	\begin{enumerate}
		\item	The curve $ \mathcal{X} $ has genus \[ g=\frac{1}{2}\left((q^{c}-2)(q^{a-1}+q^{b-1}-2)+(q^{c}-q)\right) .\]
		\item $\Div( x)=P+q^{a-1} N_{b} V-q^{a}Q$, and $\Div( y)=q^{b}P-q^{b-1} N_{a}V-Q$.

		\item $\deg(P)=q^{a-1}$,  $\deg(Q)=q^{b-1}$, and $ \deg(V)=q-1$.
		
			\item For each $\alpha \in \mathbb{F}_{q^c}^*$, there are $q^{c-1}$ elements $\beta \in \mathbb{F}_{q^c}^*$ such that $ \Tr_c(\frac{\beta}{\alpha^{q^b} })=1$, and for such pairs $(\alpha, \beta)$ there is a unique place $ D_{\alpha, \beta} $ of degree one with 
			\[ x\equiv \alpha \mod{P_{\alpha,\beta}} \text{ and } y\equiv \beta \mod{P_{\alpha,\beta}}.
			\]
			Let $ D: = \sum D_{\alpha, \beta} $. Then $ \deg(D) = q^{c-1} (q^{c}-1) $.
	\end{enumerate}
\end{prop}
 
The following proposition describes all the rational places on the curve $\mathcal{X}$.

\begin{prop}\label{prop:places}
	\begin{enumerate}
		\item There exists a unique rational place $ P_1 := (x=0, y=0, x^{-q^b} y = a^{-1})$ in the divisor $ P $ if and only if $ p \nmid  a$.
		\item There exists a unique rational place $ Q_1 := (x=\infty, y=\infty, x^{-1} y^{q^a} = b^{-1})$ in the divisor $ Q $ if and only if $ p \nmid  b$.
	 \item The divisor $ V $ can be written as  $ V = \sum _{ \mu^{q-1}=-1 } V_\mu  $, where $ V _ \mu = ( x= 0 ,y=\infty ,x^{q^{b -1 } N_{a}}y^{q^{a -1 }N_{b}} = \mu )$ denotes a rational place  in $ V $ just in case $ p =2$.
	\end{enumerate}
	Therefore, all the possible  rational places on the curve are the following: $D_{\alpha,\beta}$,
		$P_{1}$, $Q_{1}$, and $V_{\mu}$.
\end{prop}
\begin{proof}
	\begin{enumerate}
		\item Suppose that $ P_\gamma $ is 	a rational place in $ P $.  It can be deduced by the divisors of $x$ and $y$ that $ v_{P_\gamma}(x^{-q^b}y)=0 $ and $ v_{P_\gamma}(x^{-1}y^{q^a})> 0 $. We can assume that $ x^{-q^b}y \equiv \gamma \mod P_\gamma $. Taking evaluation  in Equation (\ref{eq:curve_x_y}), we have $\Tr_a \gamma = 1$. Actually, it is shown in \cite{Bassa} that the conorm of $ P $ with respect to $ \overline{\mathbb{F}}_q /\mathbb{F}_{q^c}  $ is 
		
		\[  \Con_{\overline{\mathbb{F}}_q /\mathbb{F}_{q^c}}  (P) := \sum_{ \Tr_a \gamma=1, \gamma \in \overline{ \mathbb{F}}_q} \overline{P}_\gamma,
		\]
		where $ \overline{P}_\gamma:=(x=0, y=0,x^{-q^b}y =\gamma ) $ denotes a rational place  in the function field  $ \overline{\mathbb{F}}_q (\mathcal{X}) $. Hence,  each rational place in the function field  $  {\mathbb{F}}_{q^c} (\mathcal{X}) $ can be written as $P_\gamma   =(x=0, y=0,x^{-q^b}y =\gamma ) $ with 
		\begin{equation}\label{eq:trace}
		\Tr_a \gamma= \gamma + \gamma ^{q} +  \gamma ^{q^2}+\ldots + \gamma ^{q^{a-1}}=1  \text{ and } \gamma \in \mathbb{F} _{q^c} .
		\end{equation} 
		Now we only need to show that there exist a unique solution in Equation (\ref{eq:trace}) if and only if $ p \nmid a $.
		If $\gamma$ is a solution in Equation (\ref{eq:trace}), then 
		\begin{equation} \label{eq:gamma}
		\gamma^q + \gamma ^{q^2} +  \gamma ^{q^3}+\ldots + \gamma ^{q^{a }}=1   .
		\end{equation} 
		Combining Equations (\ref{eq:trace}) and   (\ref{eq:gamma}) we find that $ \gamma ^{q^{a }}=\gamma  $ and therefore $   \gamma \in \mathbb{F} _{q^c} \cap  \mathbb{F} _{q^a} = \mathbb{F}_q $.
		Now Equation (\ref{eq:trace}) becomes $ a \gamma =1  $, and then $ \gamma =a^{-1} $ just in case $ p \nmid a $. In other words, 
		$ P_1 := (x=0, y=0, x^{-q^b} y = a^{-1})$ represent a rational place in $ \mathbb{F} _{q^c}(\mathcal{X}) $ when $ p \nmid a $.
		\item Similar to the proof of assertion 1).
		\item Suppose that $ V_\mu  $ is 	a rational place in $ V $.  It is easy to show that $ v_{V_\mu}(x^{q^{b  -1 } N_{a}}y^{q^{a -1 }N_{b}})=0 $. Let us assume that $ x^{q^{b -1 } N_{a}}y^{q^{a -1 }N_{b}} \equiv \mu  \mod V_\mu  $. Multiplying both sides of Equation(\ref{eq:curve_x_y}) with $  x^{q^{c-1}}y ^{-q^{a-1}} $, we obtain
	\[ \left(	x^{q^{b-1} N_a }y^{q^{a-1} N_b } \right)^{q-1}  +  1 = \kappa, \]
	where $ v_{V_\mu } (\kappa) > 0 $. Then we have $ \mu ^{q-1}=-1 $. Note that the equation $ \mu ^{q-1}=-1 $ has $q-1$ solutions in $ \mathbb{F}_{q^c} $ if and only if $ 2| q $. Now the assertion 3) can be   deduced similarly.
	\end{enumerate}
\end{proof}

\section {explicit bases for Riemann-Roch spaces}
In the rest of this paper, we shall always
assume that $p \nmid a$. Now by Proposition \ref{prop:places}, the rational place $P_1 $ exists. The divisor $P $ can be decomposed by $ P= P_1 + P _0 $, with $ \deg (P_1)  = 1$, and $ \deg(P_0)={q^{a-1}-1} $. We remark that the assumption $p \nmid a $ is not essential. If  $p\mid a $, then  $p \nmid b$. And we obtain a rational place $ Q_1 $ in $ Q $.
All the results in both cases are similar.

Our next aim is to   determine a basis for a   space $ \mathscr{L}( v P_1+r P_0 + s Q+tV) $ and then we can construct a generator matrix for our AG codes.  
Let $u: = 
a^{-1}- \dfrac{ {y^{q^{a}}}}{x} - \dfrac{ {y^{q}}}{x^{q^{a}}}  $, $ w :=\dfrac{y ^{q^a}}{xu}  $, and $ z :=\dfrac{y }{ x^{q^{b}}}$ be three elements in $ \mathbb{F}_{q^c} (\mathcal{X}) $. We want to construct a basis of $ \mathscr{L}( v P_1+r P_0 + s Q+tV) $ where all elements are of the form
\[ x^i z^j w^k \text{ with $( i ,j, k )\in \mathbb{Z}^3 $} .
\]
\begin{prop}\label{prop:div}
The divisors of $u$, $w$, and $z$ are given by
\begin{align*}
 \Div( u)&=
 (q^{c}-1)P_{1}- N_{c}V ,  \\
 	 \Div(z)&=-q^{b-1} N_c V + (q^c -1) Q,  \\
 \Div (w)&= (q^c -1) P_0 -(q^{a-1} -1 )N_c V .
\end{align*}

\end{prop}
\begin{proof}
	The divisor of $u$ is computed in \cite{Bassa}, the others can be calculated directly.
\end{proof}

   We denote by $ \floor{x}  $  the largest integer not greater than $ x $ and by $\ceil{x}$ the smallest integer not less than $ x $. It is easy to show that $ j = \ceil{\frac{\alpha}{\beta}} $ is equivalent to 
  \[  j \in \mathbb{Z} \quad \text{and } \alpha\leqslant \beta j < \alpha + \beta .
  \]
   Let us denote the lattice point set
  \begin{align*}
  \Omega_{v,r,s,t}:=\{ (i,j,k)|&   -v \leqslant i , \\
  &-r \leqslant  i+(q^c -1) k  <-r+(q^c -1)  ,  \\
  &-s \leqslant	-q^{a}i+(q^{c}-1)j <  (q^c -1)-s, \\
  &-t\leqslant (q^{a -1}N_{b})i-(q^{b -1 }N_{c})j -  (q^{a-1} - 1) N_{c} k    ~\},
  \end{align*} 
  	or equivalently,
  	\begin{align*}
  	\Omega_{v,r,s,t}:=\{ (i,j,k)|  
  	&~i   \geqslant -v, j= \ceil{\frac{q^a i -s}{q^c-1}},  k = \ceil{\frac{-i-r}{q^c-1}}, \\
  	&(q^{a -1}N_{b})i-(q^{b -1 }N_{c})j   + t \geqslant (q^{a-1} - 1) N_{c} k  ~\}.
  	\end{align*} 
  	
  	With these notations we have
  	the following proposition.
  \begin{prop}\label{prop:Omega_num}
	  There exists a constant $ C $ depend on the $r ,s,t$; such that, for $ v\geqslant C $, the number of $\Omega_{v,r,s,t}$ verifies
	  \[
	  \# \Omega _{v,r,s,t}=1-g+ v + (q^{a-1}-1)r +q^{b-1} s +(q-1)t.
	   \]
  \end{prop}
   We shall give a proof of Proposition \ref{prop:Omega_num} later. The following proposition is the main result of this paper which can be applied to encoding multi-point codes.
\begin{prop}\label{prop:basis_all}
	{The elements $ x^i z^j w^k $ with $ (i,j,k)\in\Omega_{v,r,s,t} $  form a basis of $ \mathscr{L}( v P_1 +r P_0 + s Q+tV) $.		
	}
\end{prop}

 \begin{proof}
 	Proposition \ref{prop:valuation} and Proposition  \ref{prop:div}  imply
 	\begin{align*}
 	\Div( x^i z^j w^k ) & =  i\Div( x)+j \Div( z) + k \Div (w)\\
 	& = iP_0 +i P_1 +q^{a-1} N_{b}i V-q^{a}i Q \\
 	& \quad -q^{b-1} N_c j V + (q^c -1) j Q  \\
 	& \quad +(q^c -1) k P_0 -(q^{a-1} -1 )N_c k V  \\
 	& =  i P_1 +\left(i + (q^c-1) k \right)P_0 + \left(-q^a i +(q^c -1 )j \right)Q  \\
 	&  \quad + \left(q^{a-1} N_{b}i -q^{b-1} N_c j-(q^{a-1} -1 )N_c k \right)V.
 	\end{align*}
 	Thus, $ x^i z^j w^k \in \mathscr{L}(vP_1 + r P_0 +sQ +tV ) $ if and only if the following conditions hold
 	\begin{align*}
 	& P_1:  -v \leqslant i  ,\\
 	& P_0: -r \leqslant  i+(q^c -1) k  ,   \\
 	& Q: -s \leqslant	-q^{a}i+(q^{c}-1)j  , \\
 	& V:-t\leqslant (q^{a -1}N_{b})i-(q^{b -1 }N_{c})j -  (q^{a-1} - 1) N_{c} k   .
 	\end{align*} 	
 	Hence, all the elements in  $\left\{ x^{i}z^{j}w^k|(i,j,k)\in\Omega_{v,r,s,t}\right\} $ are contained in $ \mathscr{L}(vP_1 + r P_0 +sQ +tV ) $. 
 	
 	Note that for $ (i,j,k)\in \Omega_{v,r,s,t} $ both $j$ and $k$ are determined by $i$ which is exactly the valuation of $ x^i z^j w^k  $ with respect to $P_1 $. In other words, all elements in $\left\{ x^{i}z^{j}w^k|(i,j,k)\in\Omega_{v,r,s,t}\right\} $ has different valuations with respect to $P_1 $, so they are linearly independent. 
 	
 	Denote by $d$ the dimension of  $ \mathscr{L}(vP_1 + r P_0 +sQ +tV ) $. To complete the proof, we only need to show that  $ d = \# \Omega_{v,r,s,t} $ for large $v $. 
 	By Riemann-Roch theorem, for large $v $, we have
 	\begin{equation}\label{eq:dim_L}
 	d  = 1-g+v+(q^{a-1}-1)r +q^{b-1}s+(q-1)t. 
 	\end{equation}
 	The proposition now follows from combining Equation (\ref{eq:dim_L}) and Propositions  \ref{prop:Omega_num}.
 \end{proof}
 
 \begin{cor}\label{cor:basis}
 	The elements $x^i y^j u^k$ with $ (i,j,k) \in \Omega'_{v,r,s,t} $ form a basis of the Riemann-Roch space $\mathscr{L}(vP_1 + rP_0 + sQ +tV)$, where $ \Omega'_{v,r,s,t}$ denotes the lattice point set 
 	\begin{align*}
 		\Omega'_{v,r,s,t}: =\{(i,j,k)|  -v & \leqslant i+ q^b j+ (q^c-1) k \\
 		-r & \leqslant i+ q^b j<-r+ (q^c-1) \\
 		-s & \leqslant -q^a i -j<  -s+(q^c-1)\\
 		-t & \leqslant q^{q-1}N_b i-q^{b-1}N_a j -N_c k ~ \}.
 	\end{align*}
 \end{cor}
 \begin{proof}
 	By definition, 
 	\[
 	x^i z^ j w ^k = x^ {i-q^b j - k} y^{ j +q^a k}  u^{-k}  .
 	 \]
 	 Let $ i' : = i-q^b j - k  $, $ j' := j +q^a k $, $ k' := -k $.
 	 Then $ \Omega_{v,r,s,t} $ becomes
\begin{align*}
	\{(i',j',k')|  -v & \leqslant i'+ q^b j'+ (q^c-1) k' \\
			     -r & \leqslant i'+ q^b j'<-r+ (q^c-1) \\
			     -s & \leqslant -q^a i' -j'<  -s+(q^c-1)\\
			     -t & \leqslant q^{q-1}N_b i' -q^{b-1}N_a j' -N_c k'~ \},
\end{align*}
which is equal to $ \Omega'_{v,r,s,t} $.	So we have
$  \{{x^i y^j u^k}| (i,j,k) \in \Omega'_{v,r,s,t} \}= \{{x^i z^j w^k}| (i,j,k) \in \Omega_{v,r,s,t}\}   $. 
 \end{proof}
 
In order to count the elements of $ \Omega_{v,r,s,t} $ we need some preparations.
\begin{lem} \label{lem:segment}
	Assume that $\alpha$ is an integer.
	\begin{enumerate}
		\item Let $L^{(1)}_\alpha: = \{ ( {i},  {j})|     
		0   \leqslant  {i}<  q^c-1 , \text{ and } q^{a-1}N_b   {i}-q^{b-1}N_c  {j} = -\alpha  \} $. Then 
		\[ \#  L^{(1)}_\alpha = \begin{cases}
		q-1  &  \text{ for $ q^{b-1} | \alpha $}, \\
		0 &\text{otherwise}.
		\end{cases}  \]
		\item Let $L^{(2)}_\alpha: = \{ ( {i},  {j})|     
		0   \leqslant  {i}<  q^c-1 , \text{ and }   -q^a  {i} +(q^c-1)   {j} = -(q^c-1)q -\alpha  \}  $. Then  $ \#  L^{(1)}_\alpha = 1 $.
		\item Let $L^{(3)}_\alpha: = \{ ( {m},  {l})|     
		0   \leqslant  {m}<  q^b-1 , \text{ and }   q^{b-1} l - N_c m =  \alpha  \}  $. Then  $ \#  L^{(3)}_\alpha = 1 $.
	\end{enumerate}
\end{lem}
\begin{proof}
We prove only the first assertion of this lemma. The other conclusions can be deduced similarly.
Note that  the greatest common
factor of $ q^{a-1}N_b $ and $ q^{b-1} N_c $ is $ \gcd ( q^{a-1}N_b ,  q^{b-1} N_c )=q^{b-1} $. The equation $ q^{a-1}N_b i -q^{b-1} N_c j =-\alpha $  has integer solutions if and only if $ p^ {b-1} $ can be divided by   $\alpha$. If $ q^{b-1} \nmid \alpha $, then $  \#  L_\alpha^{(1)} =0 $. If $ q^{b-1}\mid \alpha $,  we can assume that  $ (i_0 , j_0) $  is an  integer solution. We claim that all the other solutions are given by $( i_0 + N_c \lambda ,j_0 + q N_b \lambda )$ with $ \lambda \in \mathbb{Z}  $.  It
is easy to check that $( i_0 + N_c \lambda ,j_0 + q N_b \lambda )$ is   an  integer solution. Conversely, let $  (i,j) $ be an integer solution different from $ (i_0, j_0) $. Substituting both solutions into the equation we obtain
\[ 
q  N_b( i -i_0 ) -   N_c (j-j_0) = 0.
 \]
 So $ N_c | (i- i_0)  $ and $ qN_b | (j- j_0) $. Hence, $  L_\alpha^{(1)} $ can be written as
 \[ L_\alpha^{(1)}=\{ ( i_0 + N_c \lambda ,j_0 + q N_b \lambda ) | 0 \leqslant i_0 + N_c \lambda < q^c -1 \} .
 \]
 Hence, $  \#  L_\alpha^{(1)} =q-1 $.
\end{proof}

\begin{lem} \label{lem:Psi_num}
	Suppose that  $ 0 \leqslant m < q^ {b-1 } $, and  $s, t \geqslant 0 $. Let  $ \Psi_{m}  $ be a lattice point set 
	\begin{align*}
		\Psi_{m}  :=   \{ ( {i},  {j})|     
		   \quad  0   &\leqslant  {i}<  q^c-1  , \\
		    -t    & \leqslant q^{a-1}N_b   {i}-q^{b-1}N_c  {j}+N_c m, \\
	      -s &- (q^c-1)q      \leqslant -q^a  {i} +(q^c-1)   {j}    ~\}.   
	\end{align*}
	Then  
	 \begin{equation}\label{eq:Psi_m_st}
	 \#\Psi_{m} =  (q^c+1) q /2  + \floor{\frac{t + N_c m }{q^{b-1}}}(q-1) + s , 
	\end{equation} 
	and 
\begin{equation}\label{eq:Psi_st}
	\sum_{m=0}^ {q^{b-1}-1}  \#\Psi _{m} = \frac{1}{2} (q^{ c+b-1} + q^{c+b} - q^c +q) +q^{b-1} s +(q-1)t.
\end{equation}
\end{lem}
\begin{figure}[H]
	\centering
	\begin{tikzpicture}[scale=0.45]
 	\path [fill=orange] (0,0)--(0,-3) -- (8,2) -- (0,0);
	\draw [green] (0 ,0 )--(8  ,2 );
	\draw [green] (0,2)--(8,4);
	
	\draw [red]  (0,-6.5)--(8,-1.5);
	\draw [red]  (0,-3)--(8,2);
	\draw [blue] (0,-7)--(0,5);
	\draw [blue,dashed] (8,-7)--(8,5);
	\draw [->] (-6,0)--(-4,0);
	\draw [->] (-6,0)--(-6,2);
	\draw[color=black] node [yshift=-1ex,xshift=1ex] at (-6,2) { $ j $ };
	\draw[color=black] node [yshift=-1ex,xshift=1ex] at (-4,0) { $ i $ };
	
	\draw[color=black] node [yshift=-1ex,xshift=1ex] at (8,2) { $ B $ };
	\draw[color=black] node [yshift=-1ex,xshift=-1ex] at (0,0) { $ O $ };
	\draw[color=black] node [yshift=0ex,xshift=-1ex] at (0,-3) { $ A $ };
 
\draw[color=black] node [yshift=0ex,xshift=0ex] at (4.6,-2.5) { $ \Psi^{(2)} $ };
\draw[color=black] node [yshift=0ex,xshift=0ex] at (2.6,-0.5) { $ \Psi^{(0)} $ };
\draw[color=black] node [yshift=0ex,xshift=0ex] at (4.6, 2.5) { $ \Psi^{(1)}  $ };
	%points
 
\draw [fill] (0, -6) circle [radius=0.1];
\draw [fill] (0, -5) circle [radius=0.1];
\draw [fill] (0, -4) circle [radius=0.1];
\draw [fill] (0, -3) circle [radius=0.1];
\draw [fill] (0, -2) circle [radius=0.1];
\draw [fill] (0, -1) circle [radius=0.1];
\draw [fill] (0, 0) circle [radius=0.1];
\draw [fill] (0, 1) circle [radius=0.1];
\draw [fill] (0, 2) circle [radius=0.1];
\draw [fill] (1, -5) circle [radius=0.1];
\draw [fill] (1, -4) circle [radius=0.1];
\draw [fill] (1, -3) circle [radius=0.1];
\draw [fill] (1, -2) circle [radius=0.1];
\draw [fill] (1, -1) circle [radius=0.1];
\draw [fill] (1, 0) circle [radius=0.1];
\draw [fill] (1, 1) circle [radius=0.1];
\draw [fill] (1, 2) circle [radius=0.1];
\draw [fill] (2, -5) circle [radius=0.1];
\draw [fill] (2, -4) circle [radius=0.1];
\draw [fill] (2, -3) circle [radius=0.1];
\draw [fill] (2, -2) circle [radius=0.1];
\draw [fill] (2, -1) circle [radius=0.1];
\draw [fill] (2, 0) circle [radius=0.1];
\draw [fill] (2, 1) circle [radius=0.1];
\draw [fill] (2, 2) circle [radius=0.1];
\draw [fill] (3, -4) circle [radius=0.1];
\draw [fill] (3, -3) circle [radius=0.1];
\draw [fill] (3, -2) circle [radius=0.1];
\draw [fill] (3, -1) circle [radius=0.1];
\draw [fill] (3, 0) circle [radius=0.1];
\draw [fill] (3, 1) circle [radius=0.1];
\draw [fill] (3, 2) circle [radius=0.1];
\draw [fill] (4, -4) circle [radius=0.1];
\draw [fill] (4, -3) circle [radius=0.1];
\draw [fill] (4, -2) circle [radius=0.1];
\draw [fill] (4, -1) circle [radius=0.1];
\draw [fill] (4, 0) circle [radius=0.1];
\draw [fill] (4, 1) circle [radius=0.1];
\draw [fill] (4, 2) circle [radius=0.1];
\draw [fill] (4, 3) circle [radius=0.1];
\draw [fill] (5, -3) circle [radius=0.1];
\draw [fill] (5, -2) circle [radius=0.1];
\draw [fill] (5, -1) circle [radius=0.1];
\draw [fill] (5, 0) circle [radius=0.1];
\draw [fill] (5, 1) circle [radius=0.1];
\draw [fill] (5, 2) circle [radius=0.1];
\draw [fill] (5, 3) circle [radius=0.1];
\draw [fill] (6, -2) circle [radius=0.1];
\draw [fill] (6, -1) circle [radius=0.1];
\draw [fill] (6, 0) circle [radius=0.1];
\draw [fill] (6, 1) circle [radius=0.1];
\draw [fill] (6, 2) circle [radius=0.1];
\draw [fill] (6, 3) circle [radius=0.1];
\draw [fill] (7, -2) circle [radius=0.1];
\draw [fill] (7, -1) circle [radius=0.1];
\draw [fill] (7, 0) circle [radius=0.1];
\draw [fill] (7, 1) circle [radius=0.1];
\draw [fill] (7, 2) circle [radius=0.1];
\draw [fill] (7, 3) circle [radius=0.1];
	\end{tikzpicture}
	\protect\caption{The lattice point set $\Psi_{m} $ }
\end{figure}
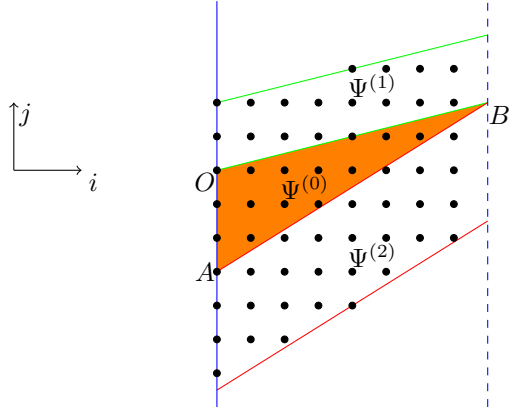

\begin{proof}
	We  split $ \Psi_{m} $ into three  parts $\Psi ^{(0)}  $, $ \Psi ^{(1)} $  and $ \Psi ^{(2)} $; namely 
	\begin{align*}
		\Psi ^{(0)} : =   \{ ( {i},  {j})|     
		 &~ 0   \leqslant  {i}<  q^c-1 , \\
		 &~ 0  \leqslant q^{a-1}N_b   {i}-q^{b-1}N_c  {j}, \\
		& -(q^c-1)q    \leqslant -q^a  {i} +(q^c-1)   {j}    ~\} ,   
	\end{align*}	
	\begin{align}
		\Psi ^{(1)} : =   \{ ( {i},  {j})|     
	&	~ 0   \leqslant  {i}<  q^c-1 , \nonumber \\
	& -t-N_c m   \leqslant q^{a-1}N_b   {i}-q^{b-1}N_c  {j} < 0, \nonumber \\
	&	-(q^c-1)q   \leqslant -q^a  {i} +(q^c-1)   {j}  ~\} , \label{eq:Psi_1_2}
	\end{align}
	and
	\begin{align*}
		\Psi ^{(2)} : =   \{ ( {i},  {j})|     
		& ~0   \leqslant  {i}<  q^c-1 , \\
		& -t-N_c m \leqslant q^{a-1}N_b   {i}-q^{b-1}N_c  {j}, \\
		& -s -(q^c-1)q    \leqslant -q^a  {i} +(q^c-1)   {j} < -(q^c-1)q   ~\}  .   
	\end{align*}
	
Then we have $ \# \Psi_{m}  = \# \Psi^{(0)} + \#\Psi^{(1)}   + \# \Psi^{(2)}  $. Equation (\ref{eq:Psi_m_st}) now follows from the following assertions
	\begin{enumerate}
		\item $ \# \Psi^{(0)}= q(q^c+1)/2 $.
		\item $ \# \Psi^{(1)}=(q-1) \floor{ \dfrac{t+N_c m}{q^{b-1}}}$.
		\item $ \# \Psi^{(2)}=s$.
	\end{enumerate}
	
	Let  $ O =( 0, 0)$, $A= (0, -q )$, and
	$B= (q^c -1, q^{b+1}-q)$. Denote 	by $S $  the area  of the triangle $ \triangle OAB $, and by $ M  $ the number of lattice points in the boundary of $ \triangle OAB $. Applying Pick's Theorem, the number $I $ of lattice points within $ \triangle OAB $ verifies
	\[
	I = S - M/2 + 1 .
	\]
	It is easy to see that $ S = q(q^c-1)/2 $.  In the same notations of Lemma \ref{lem:segment}, we have $ M = \# L_0 ^{(1)}+ \# L _0^{(2)} + q-1+1= 2q $.
	Since $  \Psi^{(0)} $ contains exactly the lattice points in the triangle $  \triangle OAB $ expect the vertex $B $, we see that
	\begin{align*}
		\# \Psi^{(0)}&=\# \triangle OAB -1= I+ M -1 \\
		&= S+ M/2 \\
		& = q(q^c+1)/2.
	\end{align*}
	This proves the first assertion. Put $ \alpha: = q^{b-1}N_c j-q^{a-1} N_b i $. For $0 \leqslant i < q^c -1 $ and $ \alpha > 0$, we have
	\begin{align*}
		-q^a i + (q^c -1 ) j & = -q^a i + \frac{(\alpha + q^{a-1} N_b i)(q-1)}{q^{b-1}} \\
		& = \frac{(q-1) \alpha - q^{a-1} i }{q^{b-1}} \\
		& > -q (q^c-1).
	\end{align*}
	So Condition (\ref{eq:Psi_1_2}) is invalid, and therefore
	\[  \Psi^{(1)} = \bigcup_{\alpha=1}^{t+ N_c m}  L _{\alpha}^{(1)}.\]
	Applying Lemma \ref{lem:segment}, we get
	\begin{align*}
		\#\Psi^{(1)}&=  \sum_{\alpha =1 }^{ t+ N_c m} \# L ^{(1)} _{\alpha}\\
		& =  (q-1) \floor{\frac{ t + N_c m  }{q^{b-1}}}  ,
	\end{align*}
	which completes the proof of the  second assertion.
	Similar to the proof above, we find that $ \#\Psi^{(2)}= s $.

	In order to complete the proof, we introduce a lattice point set 
	\[ \Phi_t : = \{(m,l)| 0 \leqslant m< q^{b-1}, 0 < l \leqslant \frac{t+ N_c m}{{q^{b-1}}}  \} .
	\]
	Then 
	\begin{equation}\label{eq:Phi_t_1}
		\# \Phi_t =  \sum_{m=0}^ {q^{b-1}-1}  \floor{\frac{t + N_c m }{q^{b-1}}}.
	\end{equation}
 Note that  
 \begin{equation}\label{eq:Psi_t_1}
	\Phi_{t} = \Phi_{0 } \bigcup  \left( \bigcup_{\alpha = 1}^{t} \# L^{(3)}_\alpha  \right).
\end{equation}
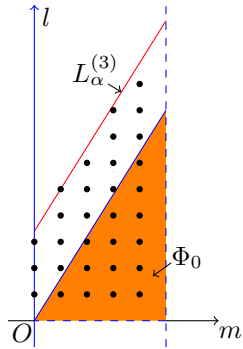
\begin{figure}[H]
	\centering
	\begin{tikzpicture}[scale=0.35]
	\path [fill=orange] (0,0)--(5,0) -- (5,8) -- (0,0);
	\draw [blue] (0 ,0 )--(5,8);
	\draw [red] (0,3.4)--(5,11.4);
	
	\draw [blue,->] (0,-1)--(0,12);
	\draw [blue,dashed] (5,-1)--(5,12);
	\draw [blue,dashed] (5,0)--(0,0);
	\draw [->] (5,0)--(7,0);
	\draw (-1,0)--(0,0);
	\draw[color=black] node [yshift=-1ex,xshift=1ex] at (0,12) { $ l $ };
	\draw[color=black] node [yshift=-1ex,xshift=1ex] at (7,0) { $ m $ };
	
	\draw[color=black] node [yshift=-1ex,xshift=-1ex] at (0,0) { $ O $ };
    \draw[color=black] node   at (5.8,2.4) { $ \Phi_0 $ };
    \draw [->] (5.2,2.2)--(4.5,1.7);
    \draw[color=black] node [yshift=-1ex,xshift=-1.5ex] at (3,10) { $ L_\alpha^{(3)} $ };
    \draw [->] (2.8,9)--(3.3,8.7);
	%points

		\draw [fill] (0, 1) circle [radius=0.1];
		\draw [fill] (1, 1) circle [radius=0.1];
		\draw [fill] (2, 1) circle [radius=0.1];
		\draw [fill] (3, 1) circle [radius=0.1];
		\draw [fill] (4, 1) circle [radius=0.1];
\draw [fill] (0, 2) circle [radius=0.1];
\draw [fill] (1, 2) circle [radius=0.1];
\draw [fill] (2, 2) circle [radius=0.1];
\draw [fill] (3, 2) circle [radius=0.1];
\draw [fill] (4, 2) circle [radius=0.1];
\draw [fill] (0, 3) circle [radius=0.1];
\draw [fill] (1, 3) circle [radius=0.1];
\draw [fill] (2, 3) circle [radius=0.1];
\draw [fill] (3, 3) circle [radius=0.1];
\draw [fill] (4, 3) circle [radius=0.1];

		\draw [fill] (1, 4) circle [radius=0.1];
		\draw [fill] (2, 4) circle [radius=0.1];
		\draw [fill] (3, 4) circle [radius=0.1];
		\draw [fill] (4, 4) circle [radius=0.1];
 
		\draw [fill] (1, 5) circle [radius=0.1];
		\draw [fill] (2, 5) circle [radius=0.1];
		\draw [fill] (3, 5) circle [radius=0.1];
		\draw [fill] (4, 5) circle [radius=0.1];

		\draw [fill] (2, 6) circle [radius=0.1];
		\draw [fill] (3, 6) circle [radius=0.1];
		\draw [fill] (4, 6) circle [radius=0.1];
		
		\draw [fill] (3, 7) circle [radius=0.1];
		\draw [fill] (4, 7) circle [radius=0.1];
		\draw [fill] (3, 8) circle [radius=0.1];
		\draw [fill] (4, 8) circle [radius=0.1];
	    \draw [fill] (4, 9) circle [radius=0.1];
	\end{tikzpicture}
	\protect\caption{The lattice point set $\Phi_{t} $ }
\end{figure}

	Using Pick's Theorem, it is easy to show that 
    \begin{equation}
    \# \Phi_0 = (N_c -1 )(q^{b-1}-1)/2 .
    \end{equation}
	 By Lemma \ref{lem:segment} it follows that
	 \begin{equation}
	 \sum_{\alpha = 1}^{t} \# L^{(3)}_\alpha =t .
	 \end{equation}
	 Hence, we deduce that
	\begin{align} 
		\# \Phi _t & = \# \Phi_0  + \sum_{\alpha = 1}^{t} \# L^{(3)}_\alpha \nonumber \\
		& =  (N_c-1)(q^{b-1}-1)/2 +t .\label{eq:Phi_t_2}
	\end{align}
	If we combine  Equations (\ref{eq:Phi_t_1}) and (\ref{eq:Phi_t_2}), then
	we get

\begin{equation}\label{eq:sum_m}
	\sum_{m=0}^ {q^{b-1}-1}  \floor{\frac{t + N_c m }{q^{b-1}}} =  (N_c-1)(q^{b-1}-1)/2 +t.
\end{equation}
 Using Equations (\ref{eq:Psi_m_st}) and (\ref{eq:sum_m}) we obtain
	\[ 
	\sum_{m=0}^ {q^{b-1}-1}  \#\Psi^{m}_{s,t} = \frac{1}{2} (q^{ c+b-1} + q^{c+b} - q^c +q) +q^{b-1} s +(q-1)t,
	\]
	which completes the proof.
\end{proof}

\begin{lem} \label{lem:Omega_main}
	Let $\Omega_{v,0,s,t}$ be the lattice point set $\Omega_{v,r,s,t}$  with $r=0$. For $ v\geqslant v_0:= (q^c-1)(q^b +  q^{b-1} -1 )  $,  and $ s, t \geqslant 0 $,  we have 
	
	\[ \# \Omega_{v,0,s,t}  
	=   1-g+v+q^{b-1}s +(q-1)t .\]
	
\end{lem}

\begin{proof}
	In order to calculate the number of the set $ \Omega_{v,0,s,t} $, we fix the index $ k  $, and define
	\[
	\Theta_{k}:=  \{(i,j)| (i,j,k) \in \Omega_{v,0,s,t} \}.
	\]  
	Precisely  speaking, 
	\begin{align}
	\Theta_{k}:=\{ (i,j )|    -v &  \leqslant i ,\label{eq:Omega_1}\\ 
	0 & \leqslant  i+(q^c -1) k  <  q^c -1   ,  \label{eq:Omega_2}\\
	-s  & \leqslant	-q^{a}i+(q^{c}-1)j <  (q^c -1)-s,\label{eq:Omega_3} \\
	-t &\leqslant (q^{a -1}N_{b})i-(q^{b -1 }N_{c})j -  (q^{a-1} - 1) N_{c} k    ~\}. \label{eq:Omega_4}
	\end{align} 
	Then we have
	\[ \Omega_{v,0,s,t}=     \left(\bigcup_{k=-\infty }^{ q^b+q^{b-1}-1} \Theta_k \right) \bigcup\left(\bigcup_{k=q^b+q^{b-1}}^{\infty} \Theta_k \right) .\]
	We count it by two steps.
	\begin{enumerate}
		\item $\# \left( \bigcup_{k=q^b+q^{b-1}}^{\infty} \Theta_k \right) = v-v_0 $.

		Let $ (i,j)\in \Theta_k $, and $ M : = -q^a i + (q^c-1) j $. Inequality (\ref{eq:Omega_3}) tells us  $  M < (q^c-1) -s \leqslant q^c-1 $. It follows from Inequality (\ref{eq:Omega_2}) that $-i > (q^c-1)(k-1)$. For $ k \geqslant q^b + q^{b-1}$, we have
		\begin{align*}
		(q^{a -1}N_{b})i& -(q^{b -1 }N_{c})j-  (q^{a-1} - 1) N_{c} k   \\
		& = (q^{a -1}N_{b})i  -q^{b -1 }\frac{ M + q^a i  }{q-1}  -  (q^{a-1} - 1) N_{c} k   \\
		& = -\frac{q^{a-1} i }{q-1} - \frac{q^{b-1}M}{q-1} - (q^{a-1}-1)N_c k \\
		& >  \frac{q^{a-1} (q^c-1)(k-1) }{q-1} - \frac{q^{b-1}(q^c-1)}{q-1} - (q^{a-1}-1)N_c k \\
		& =  N_c (k- q^b -q^{b-1}) \\
		& \geqslant - t,
		\end{align*} 
		which means that Condition (\ref{eq:Omega_4}) is invalid. We claim that 
\begin{equation}\label{eq:combine}
		 \bigcup_{k=q^b+q^{b-1}}^{\infty} \Theta_k  = \left\lbrace (i,j) \left| -v \leqslant i < -v_0, j = \ceil{\frac{q^a i -s}{q^c-1}} \right. 
		\right\rbrace .
\end{equation} 
		To see this, we shall write down Condition (\ref{eq:Omega_2}) for various $k $. We put $v_\mu := (q^c-1)(q^b+q^{b-1}-1+\mu )$ for $ \mu \in \mathbb{N} $.  Then 
		\begin{align*}
		         -v_1\leqslant i& < -v_0  \quad \text{for } k=q^b+q^{b-1}, \\
		         -v_2\leqslant i& < -v_1  \quad \text{for } k=q^b+q^{b-1}+1, \\
		          -v_3\leqslant i& < -v_2  \quad \text{for } k=q^b+q^{b-1}+2, \\
		         & \ldots \\
	             -v_{\mu +1 }\leqslant i& < -v_\mu  \quad \text{for } k=q^b+q^{b-1}+ \mu. 
		\end{align*}
		Combining Condition (\ref{eq:Omega_2}) for $ k \geqslant q^b+q^{b-1} $, we get $ i < -v_0 $. We note that  Conditions  (\ref{eq:Omega_1}) and    (\ref{eq:Omega_3}) are  independent of $k$. Then     $ \bigcup_{k=q^b+q^{b-1}}^{\infty} \Theta_k $ becomes
			\begin{align*}
			 \{ (i,j )|    -v &  \leqslant i , \\ 
			   i  & <  v_0 ,   \\
			-s  & \leqslant	-q^{a}i+(q^{c}-1)j <  (q^c -1)-s ~\},
			\end{align*}
			which implies Equation (\ref{eq:combine}).
		By Equation (\ref{eq:combine}) one can easily verify that $ \# \left( \bigcup_{k=q^b+q^{b-1}}^{\infty} \Theta_k \right) = v-v_0 $.

		\item $	\# \left( \bigcup_{k=-\infty }^{ q^b+q^{b-1}-1}  \Theta_k \right) =  1-g+v_0+q^{b-1}s +(q-1)t.$

		It is important to write   $ k = q^ {b-1} l + m $ with $ 0 \leqslant m \leqslant q^{b-1} -1  $ and $ l \leqslant q $.
		Let $\widetilde{i} : = i+(q^c-1)k  $,    and $ \widetilde{j}= j+ q^a k - l $. Then $ \Theta_{k} $  becomes   
		\begin{align}
		{\widetilde{\Theta}}_{l,m}: =\{ (\widetilde{i}, \widetilde{j})|  \quad \widetilde{i} & \geqslant   -v+(q^c-1)k, \label{eq:tilde_1}\\
		\quad 0 & \leqslant \widetilde{i}<  q^c-1 ,\label{eq:tilde_2}\\
		-s & -(q^c-1) l    \leqslant -q^a \widetilde{i} +(q^c-1) \widetilde{j}  <-s  - (q^c-1) ( l-1) ,\label{eq:tilde_3}\\
		-t  & \leqslant q^{a-1}N_b \widetilde{i}-q^{b-1}N_c \widetilde{j}+N_c m  ~ \}\label{eq:tilde_4},
		\end{align}
		and then
	\begin{equation*}
		\bigcup_{k= -\infty }^{q^b+q^{b-1}-1} \Theta_k  =\bigcup_{ m =0 }^{ q^{b-1}-1} 	\bigcup_{ l = -\infty }^{q } \widetilde{\Theta}_{l,m}.
	\end{equation*}
	 Set $ \Psi _{m}:= \bigcup_{l=-\infty }^{q }  {\widetilde{\Theta}}_{l,m} $, we have
	 	\begin{equation*}
	 	\# \left( \bigcup_{k=-\infty }^{ q^b+q^{b-1}-1}  \Theta_k \right)   =\sum_{m=0 }^{ q^{b-1}-1}   \# \Psi_{m}. 
	 	\end{equation*}	 
		Note that Inequality (\ref{eq:tilde_1}) is invalid since  $    -v + (q^c-1)k \leqslant -v + v_0 \leqslant 0  $.  	
		 We see that  Conditions  (\ref{eq:tilde_2}) and    (\ref{eq:tilde_4}) are  independent of $l$, so we can combine Condition (\ref{eq:tilde_3}) for $ l \leqslant q$. Let $ s_{\mu } :=  s + (q^c -1 )(q  - \mu )$ for $ \mu  \in \mathbb{N}$. Then Condition   (\ref{eq:tilde_3}) can be expressed as
				\begin{align*}
				-s_0 \leqslant -q^a \widetilde{i} +(q^c-1) \widetilde{j} & < -s_1  \quad \text{for }  l=q, \\
				-s_1\leqslant -q^a \widetilde{i} +(q^c-1) \widetilde{j} & < -s_2  \quad \text{for } l=q-1,  \\
				-s_2\leqslant -q^a \widetilde{i} +(q^c-1) \widetilde{j} & < -s_3  \quad \text{for } l=q-2,  \\
				& \ldots \\
				-s_{\mu  }\leqslant -q^a \widetilde{i} +(q^c-1) \widetilde{j} & < -s_{\mu+1 } \quad \text{for } l=q- \mu. 
				\end{align*}
				This give a total condition $ -s_0\leqslant -q^a \widetilde{i} +(q^c-1) \widetilde{j}$.
		 So $ \Psi _{m} $ can be rewritten as
		\begin{align*}
		\Psi _{m} =   \{ (\widetilde{i}, \widetilde{j})|     
		 \quad 0 &  \leqslant \widetilde{i}<  q^c-1 , \\
		-t  & \leqslant q^{a-1}N_b \widetilde{i}-q^{b-1}N_c \widetilde{j}+N_c m, \\
	 -s	&-(q^c-1)q     \leqslant -q^a \widetilde{i} +(q^c-1)  \widetilde{j}    ~\}.  \label{eq:cup}
		\end{align*}
  Then we have by Lemma \ref{lem:Psi_num} that 
		\begin{align*}
		\# \left( \bigcup_{k=-\infty }^{ q^b+q^{b-1}-1}  \Omega_k \right) & = \sum_{m=0 }^{q^{b-1}-1} \# \Psi_{m} \\
		& = \frac{1}{2} (q^{ c+b-1} + q^{c+b} - q^c +q) +q^{b-1} s +(q-1)t \\
		& = 1-g+v_0 +q^{b-1}s +(q-1)t. 
		\end{align*}

	\end{enumerate}
	In summary, 
	\[ \# \Omega_{v,0,s,t}= 1-g+v +q^{b-1}s +(q-1)t. 
	\]
	This completes the proof.
\end{proof}

	\begin{lem}\label{lem:div_equi}
		If $s+ q^a r = \widehat{s}+\sigma (q^c-1) $, $ t -q^{a-1}N_b -(q^c -1) \sigma = \widehat{t} + N_c \lambda $, then 
    	\[ 	vP_1 +r P_0 +s Q +tV \sim \left( v + (q^c-1)\lambda -r \right) P_1 +\widehat{s} Q +\widehat{t} V .\]
	\end{lem}
	\begin{proof}
		By direct computation, we have
		\begin{align*}
		 vP_1 &+r P_0 +s Q +tV+\Div(x^{(q^c-1 )\lambda -r }z^{q^a \lambda -\sigma }w^{-\lambda})\\
		 & =vP_1  +r P_0 +s Q +tV +\left({(q^c-1 )\lambda -r }\right) \Div(x) \\
		 & \quad +  \left({q^a \lambda -\sigma }\right) \Div(z) -\lambda \Div(w) \\
		 &=	 vP_1  +r P_0 +s Q +tV \\
		 & \quad + \left({(q^c-1 )\lambda -r }\right) \left(P+q^{a-1} N_{b} V-q^{a}Q \right)\\
		 & \quad + \left({q^a \lambda -\sigma }\right) \left(-q^{b-1} N_c V + (q^c -1) Q \right)\\
		  & \quad - \lambda \left((q^c -1) P_0 -(q^{a-1} -1 )N_c V \right)\\
		  & =\left( v + (q^c-1)\lambda -r \right) P_1 + \widehat{s} Q +\widehat{t} V.
		\end{align*}
	\end{proof}
	\begin{lem}\label{lem:Omega_equi}
		Suppose that $s+ q^a r = \widehat{s}+\sigma (q^c-1) $ with $ 0 \leqslant \widehat{s}< q^c-1 $, and  $ t -q^{a-1}N_b -(q^c -1) \sigma = \widehat{t} + N_c \lambda $ with $ 0 \leqslant t < N_c $. Let $\widehat{v}:= v + (q^c-1)\lambda -r  $. Then 
	\[ 	\# \Omega_{v,r,s,t}= \# \Omega_{\widehat v,0,\widehat s,\widehat t}. \]
	\end{lem}
	\begin{proof}
		The proof of Lemma \ref{lem:div_equi} leads us to make a transformation $i = \widehat{i} + (q^c-1)\lambda -r $, $j=\widehat{j}+ q^a \lambda - \sigma $, and $k = \widehat{k} -\lambda $.  So we obtain 
			\begin{align*}
			\widehat{\Omega}_{v,r,s,t}:=\{ (\widehat{i},\widehat{j},\widehat{k})|   -\widehat{v}  &\leqslant \widehat{i}  ,\\
			 0 &\leqslant  \widehat{i}+(q^c -1) \widehat k  <0+(q^c -1)  ,  \\
			 -\widehat{s} &\leqslant	-q^{a}\widehat{i}+(q^{c}-1)\widehat j <  (q^c -1)-\widehat{s}, \\
			 -\widehat{t} & \leqslant (q^{a -1}N_{b})\widehat{i}-(q^{b -1 }N_{c}) \widehat{j} -  (q^{a-1} - 1) N_{c} \widehat{k}    ~\},
			\end{align*} 
			which implies the lemma.
	\end{proof}

	Now  Proposition \ref{prop:basis_all} follows easily from Lemmas \ref{lem:Omega_main} and \ref{lem:Omega_equi}.

\section{the properties of the codes}
 In this section, we study the linear code
 \[
 C_{v,r,s,t}=C_{\mathscr{L}}(  D,  vP_1 + rP_0 + s Q+t V).
 \] 
 The length of $C_{v,r,s,t}$ is $n:=\deg(D)=(q^{c}-1)q^{c-1}$. For convenience we set $ G:= vP_1 + rP_0 + s Q+t V $ with 
 \[ 
  \deg( G) = v + (q^{a-1}-1)r+ q^{b-1 }s+(q-1)t  .
  \]
 It is well known that the dimension of an AG code $ C_{\mathscr{L}}(  D, G) $ is given by
 \begin{equation}
 \dim C_{\mathscr{L}}(  D, G)=\dim\mathscr{L}(G)-\dim \mathscr{L}(G-D).\label{eq:dim}
 \end{equation}
Set $ R: = n+2g-2 $. If $\deg(G) >R$, then the Riemann-Roch Theorem and Equation (\ref{eq:dim}) yield 
  \begin{align*}
  \dim C_{v,r,s,t}&=(1-g+\deg(G)) -(1-g+\deg(G-D))\\
  &= \deg D = n,
  \end{align*}
 which is trivial. So we should only consider the case $0 \leqslant \deg(G)<R $. 
  
\begin{defn}
	 Two codes $ C_1, C_2 \subseteq \mathbb {F}^{n}_{q^c} $ are said to be \textbf{equivalent}  if there is a vector $ a =(a_1, a_2, \ldots, a_n ) \in (\mathbb {F}_{q^c}^{*})^n $ such that $ C_2 = a \cdot C_1 $; i.e.,
	 \[
	 C_2 = \left\{(a_1 c_1, a_2 c_2, \ldots , a_n c_n)|(c_1, c_2,\ldots, c_n)\in C_1\right\}.
	 \] 
	 Denote by $ C^\bot $ the dual of $ C $. The code $ C $ is called \textbf{self-dual} (resp. \textbf{self-orthogonal}) if $ C =  C ^\bot$ (resp. $ C \subseteq  C ^\bot$).  
\end{defn}

\begin{prop}[\cite{Stichtenoth}]\label{prop:equivalent}
	Suppose $G_{1}$ and $G_{2}$ are divisors with $G_{1}\sim G_{2}$
	and $\supp G_{1}\cap \supp D=\supp G_{2}\cap \supp D=\emptyset$, then $C_{\mathscr{L}}(  D, G_{1})$
	and $C_{\mathscr{L}}(  D, G_{2})$ are equivalent. 
\end{prop}

\begin{prop}\label{prop:equivalent_vst}
	Suppose that $s+ q^a r = \widehat{s}+\sigma (q^c-1) $ with $ 0 \leqslant \widehat{s}< q^c-1 $, and  $ t -q^{a-1}N_b -(q^c -1) \sigma = \widehat{t} + N_c \lambda $ with $ 0 \leqslant t < N_c $. Let $\widehat{v}:= v + (q^c-1)\lambda -r  $. Then the code $C_{v, r,s,t}$ is equivalent to  $  C_{\widehat{v} ,0,  \widehat s , \widehat t  } $.
\end{prop}
\begin{proof}
	
	It follows easily from Lemma \ref{lem:div_equi} and Proposition \ref{prop:equivalent}.  
\end{proof}

 We use the following lemma to calculate the dual of $ C_{v,r,s,t} $.
 \begin{lem}[\cite{Stichtenoth}]\label{lem:dual}
 	Let $\tau $ be an element of the function field of $ \mathcal{X}$ such that $ v_{P_i}(\tau )=1$ for all rational places $ P_i $ contained in the divisor $D$. Then the dual of $ C_{\mathscr{L}}(  D, G)$ is
 	
 	\[ C_{\mathscr{L}}(  D, G)^{\bot} =C_{\mathscr{L}}(  D,   D-G+\Div(d\tau ) - \Div \tau ).
 	\] 
 \end{lem}

\begin{prop}\label{prop:dual}
	\label{prop:The-dual-of}The dual of $C_{v,r,s,t}$ is
		\[ C_{v,r,s,t}^{\bot}=C_{-1-v,-1-r,A-s,B- t}, \]
	where $A= q^{c+a} +q^{c}-q^{a}-2 $, and $ B=(q^{a-1}-1)N_c  -1  $.
\end{prop}
\begin{proof}
	Consider the element
	\[
	\tau :=\prod_{\alpha\in\mathbb{F}_{q^c}}(x-\alpha)=x^{q^{c}}-x.
	\]
    Then $\tau $ is a prime element for all places $ D_{\alpha, \beta} $, and its  divisor is 
	\begin{align*}
		\Div( \tau ) & =  \Div_{0}(x)+D-q^{c}\Div_{\infty}(x)\\
		& =  P  +q^{a-1}N_b V+D-q^{c+a}Q.
	\end{align*}
	It follows from \cite{Bassa} that
	\[ \Diff(F/K(x))=(q^{c}+q^{a}-2)Q+\left( (q^{a-1}-1)N_c +(q^{a-1} N_b -1) \right)V . \]
	So the divisor of  $ d\tau  $ is
	\begin{align*}
		\Div( d\tau ) & =  \Div( -dx)=-2\Div_{\infty}(x)+\Diff(F/K(x))\\
		& =  -2 q^a Q + (q^{c}+q^{a}-2)Q+\left( (q^{a-1}-1)N_c +(q^{a-1} N_b -1) \right)V \\
		& =  (q^{c}-q^{a}-2)Q+\left( (q^{a-1}-1)N_c +(q^{a-1} N_b -1) \right) V.
	\end{align*}
  Let $ \eta := d\tau /\tau  $ be a Weil differential. Set $ A:= q^{c+a} +q^{c}-q^{a}-2 $, and $ B:=(q^{a-1}-1)N_c  -1  $. The divisor of $ \eta $ is
	\begin{align*}
	\Div (\eta) &= \Div(d\tau )-\Div(\tau )\\
		& =  (q^{c}-q^{a}-2)Q+\left( (q^{a-1}-1)N_c +(q^{a-1} N_b -1) \right) V
		\\
		& \quad   -P  -q^{a-1}N_bV-D+q^{c+a}Q \\
		& = -P  +\left( (q^{a-1}-1)N_c  -1) \right) V-D+(q^{c+a} +q^{c}-q^{a}-2 )Q\\
		&   = -P  -D + A Q + B V.
	\end{align*}
	By Lemma \ref{lem:dual} the dual  of $ C_{v,r,s,t} $ is 
	\begin{align*}
	    C_{v,r,s,t}^{\bot}   &=C_{\mathscr{L}}\left(D,D -v P_1- r P_0 -s Q -tV +\Div(\eta)\right)\\
		& =  C_\mathscr{L}\left(  D,(-1-v)P_1+(-1-r )P_0   + (A-s) Q+ \left( B-t \right) V \right)\\
		& =  C_{-1-v,-1-r,A-s,B- t} .
	\end{align*}

	\end{proof}

\begin{prop}\label{prop:dimension}
	{{Suppose that}} {$0 \leqslant \deg(G)<R $. Then the following holds:}\end{prop}
\begin{enumerate}
	\item The dimension of $C_{r}$ is given by 
	\[
	\dim C_{v,r,s,t}=\begin{cases}
	\#\Omega_{v,r,s,t} &\text{for \ensuremath{0\leqslant \deg(G) <n}},\\
	 n-\#\Omega_{v,r,s,t}^{\bot}  &\text{for $ n \leqslant \deg (G) \leqslant R $}.
	\end{cases}
	\]	
	where   $ \Omega_{v,r,s,t}^{\bot}: =\Omega_{-1-v,-1-r,A-s,B- t} $.
	
	\item The minimum distance $d$ of $C_{r}$ satisfies $d\geqslant n- v- (q^{a-1}-1)r- q^{b-1}s - (q-1)t $.

	\end{enumerate}
\begin{proof}
	
	\begin{enumerate}
		\item For  $ {0\leqslant \deg(G) <n} $,  we have by Proposition \ref{prop:basis_all} and Equation (\ref{eq:dim})  that
		 \[ \dim C_{v, r,s,t}=\dim \mathscr{L}(G)= \#\Omega_{v,r,s,t}  .\]
		For $ n \leqslant \deg (G) \leqslant R $,
		Proposition \ref{prop:The-dual-of} yields 
		\[ \dim C_{v,r,s,t}=n - \dim C^{\bot}_{v,r,s,t}=n -\#\Omega_{v,r,s,t}^{\bot}.  \]
		\item The inequality follows from Goppa bound.
	\end{enumerate}
\end{proof}
By Proposition \ref{prop:basis_all} or Corollary \ref{cor:basis}, one can easily specify a generator matrix for the code $ C_{v,r,s,t} $. We fix an ordering of the set 
\begin{equation}
T:= \left\lbrace { (\alpha, \beta) \in \mathbb{F}_{q^c}^* \times \mathbb{F}_{q^c}^*  \left| \Tr_{b}(\frac{\beta^{q^ a}}{\alpha})+\Tr_{a}(\frac{\beta}{\alpha^{q^b}})=1 \right. }  \right\rbrace .
\end{equation}
For $ (i,j,k) \in \mathbb{Z}^3 $ we define the vector
\begin{equation}
E_{i,j,k} :=\left\lbrace \left. \alpha^i \beta^j \left( 
a^{-1}- \frac{ {\beta^{q^{a}}}}{\alpha} - \frac{ {\beta^{q}}}{\alpha^{q^{a}}} \right)^k  \right| (\alpha, \beta )\in T \right\rbrace \in \mathbb{F}_{q^c}^n.
\end{equation}
	\begin{prop}\label{prop:matrix}
		Suppose that  $ 0 \leqslant \deg (G) < n $. Let $ m:=\dim C_{v,r,s,t} $  and $ (i_\lambda, j_\lambda , k_\lambda) $ with $ 1\leqslant \lambda \leqslant m$ be all elements in $ \Omega'_{v,r,s,t} $. Then the $ m \times n $ matrix whose rows are $ E_{i_1, j_1, k_1}, \ldots, E_{i_m , j_m , k_m } $, is the generator matrix of $ C_{v,r,s,t} $.
	\end{prop}
	\begin{proof}
	Corollary \ref{cor:basis}.
	\end{proof}
\begin{example}
	Let us consider the case $ q = 2 $, $c =5 $, $a=3$, and $b=2$, then $n =496 $, and $g=75 $. By Proposition \ref{prop:equivalent_vst} we should only consider the codes $ C_{v,r,s,t} $ with $ r = 0, 0\leqslant s < 31, 0 \leqslant t < 31$. Applying Proposition \ref{prop:dual}, we find that  the dual code of $C_{v,r,s,t} $ is $C_{-1-v, -1-r, 278-s,   92-t} $.
	Using Proposition \ref{prop:dimension},  we can determine the dimension and the Goppa bound for $ C_{v,r,s,t} $. The GV bound is the best lower bound which is known from elementary coding theory. However, its proof is not constructive. It does not provide a simple algebraic algorithm for the construction of good long codes. While our codes $ C_{v,r,s,t} $ can be constructed explicitly by applying Proposition \ref{prop:matrix}, and it turns out that the Goppa bound of $ C_{v,r,s,t} $ improves the GV bound in a certain interval, see Figure \ref{fig:GV_bound}.  For instance, we find that the code $ C_{324,0,0,0}$ is a $ [496,250,\geqslant 172]$-code over $ \mathbb{F}_{32} $, which   oversteps  the GV bound.
			\begin{figure}[H] 	
				\centering
				\begin{tikzpicture}
				\begin{axis}[ymin=0,ymax=496,xmin=0,xmax=496, enlargelimits=false,
				xlabel=dimension,
				ylabel=distance]
				\draw [->](180,180)--(200,200);
				\draw[color=black] node   at (120,180) {GV Bound};
				\draw [->](200,270)--(180,240);
				\draw[color=black] node  at (250,280) {Goppa Bound};
				
				\addplot[red,only marks,mark options={
					fill=red,scale=0.08},mark=o] coordinates {
					(1, 452)
					(1, 453)
					(1, 454)
					(1, 455)
					(1, 456)
					(1, 457)
					(1, 458)
					(1, 459)
					(1, 460)
					(1, 461)
					(1, 462)
					(1, 463)
					(1, 464)
					(1, 465)
					(1, 466)
					(1, 467)
					(1, 468)
					(1, 469)
					(1, 470)
					(1, 471)
					(1, 472)
					(1, 473)
					(1, 474)
					(1, 475)
					(1, 476)
					(1, 477)
					(1, 478)
					(1, 479)
					(1, 480)
					(1, 481)
					(1, 482)
					(1, 483)
					(1, 484)
					(1, 485)
					(1, 486)
					(1, 487)
					(1, 488)
					(1, 489)
					(1, 490)
					(1, 491)
					(1, 492)
					(1, 493)
					(1, 494)
					(1, 495)
					(1, 496)
					(2, 442)
					(2, 443)
					(2, 444)
					(2, 445)
					(2, 446)
					(2, 447)
					(2, 448)
					(2, 449)
					(2, 450)
					(2, 451)
					(2, 452)
					(2, 453)
					(2, 454)
					(2, 455)
					(2, 456)
					(2, 457)
					(2, 458)
					(2, 459)
					(2, 460)
					(2, 461)
					(2, 462)
					(2, 463)
					(2, 464)
					(2, 465)
					(2, 466)
					(2, 467)
					(2, 468)
					(2, 469)
					(2, 470)
					(2, 471)
					(2, 472)
					(2, 473)
					(2, 474)
					(2, 475)
					(2, 476)
					(2, 477)
					(2, 478)
					(2, 479)
					(2, 480)
					(3, 439)
					(3, 440)
					(3, 441)
					(3, 442)
					(3, 443)
					(3, 444)
					(3, 445)
					(3, 446)
					(3, 447)
					(3, 448)
					(3, 449)
					(3, 450)
					(3, 451)
					(3, 452)
					(3, 453)
					(3, 454)
					(3, 455)
					(3, 456)
					(3, 457)
					(3, 458)
					(3, 459)
					(3, 460)
					(3, 461)
					(3, 462)
					(3, 463)
					(3, 464)
					(3, 465)
					(3, 466)
					(3, 467)
					(3, 468)
					(3, 469)
					(3, 470)
					(3, 471)
					(3, 472)
					(3, 473)
					(3, 474)
					(3, 475)
					(3, 476)
					(4, 433)
					(4, 434)
					(4, 435)
					(4, 436)
					(4, 437)
					(4, 438)
					(4, 439)
					(4, 440)
					(4, 441)
					(4, 442)
					(4, 443)
					(4, 444)
					(4, 445)
					(4, 446)
					(4, 447)
					(4, 448)
					(4, 449)
					(4, 450)
					(4, 451)
					(4, 452)
					(4, 453)
					(4, 454)
					(4, 455)
					(4, 456)
					(4, 457)
					(4, 458)
					(4, 459)
					(4, 460)
					(4, 461)
					(4, 462)
					(4, 463)
					(4, 464)
					(5, 430)
					(5, 431)
					(5, 432)
					(5, 433)
					(5, 434)
					(5, 435)
					(5, 436)
					(5, 437)
					(5, 438)
					(5, 439)
					(5, 440)
					(5, 441)
					(5, 442)
					(5, 443)
					(5, 444)
					(5, 445)
					(5, 446)
					(5, 447)
					(5, 448)
					(5, 449)
					(5, 450)
					(5, 451)
					(5, 452)
					(5, 453)
					(5, 454)
					(5, 455)
					(5, 456)
					(5, 457)
					(5, 458)
					(5, 459)
					(5, 460)
					(5, 461)
					(6, 429)
					(6, 430)
					(6, 431)
					(6, 432)
					(6, 433)
					(6, 434)
					(6, 435)
					(6, 436)
					(6, 437)
					(6, 438)
					(6, 439)
					(6, 440)
					(6, 441)
					(6, 442)
					(6, 443)
					(6, 444)
					(6, 445)
					(6, 446)
					(6, 447)
					(6, 448)
					(6, 449)
					(6, 450)
					(6, 451)
					(6, 452)
					(6, 453)
					(6, 454)
					(6, 455)
					(6, 456)
					(6, 457)
					(7, 427)
					(7, 428)
					(7, 429)
					(7, 430)
					(7, 431)
					(7, 432)
					(7, 433)
					(7, 434)
					(7, 435)
					(7, 436)
					(7, 437)
					(7, 438)
					(7, 439)
					(7, 440)
					(7, 441)
					(7, 442)
					(7, 443)
					(7, 444)
					(7, 445)
					(7, 446)
					(7, 447)
					(7, 448)
					(7, 449)
					(7, 450)
					(7, 451)
					(7, 452)
					(7, 453)
					(8, 424)
					(8, 425)
					(8, 426)
					(8, 427)
					(8, 428)
					(8, 429)
					(8, 430)
					(8, 431)
					(8, 432)
					(8, 433)
					(8, 434)
					(8, 435)
					(8, 436)
					(8, 437)
					(8, 438)
					(8, 439)
					(8, 440)
					(8, 441)
					(8, 442)
					(8, 443)
					(8, 444)
					(8, 445)
					(8, 446)
					(8, 447)
					(8, 448)
					(9, 421)
					(9, 422)
					(9, 423)
					(9, 424)
					(9, 425)
					(9, 426)
					(9, 427)
					(9, 428)
					(9, 429)
					(9, 430)
					(9, 431)
					(9, 432)
					(9, 433)
					(9, 434)
					(9, 435)
					(9, 436)
					(9, 437)
					(9, 438)
					(9, 439)
					(9, 440)
					(9, 441)
					(9, 442)
					(9, 443)
					(9, 444)
					(9, 445)
					(10, 420)
					(10, 421)
					(10, 422)
					(10, 423)
					(10, 424)
					(10, 425)
					(10, 426)
					(10, 427)
					(10, 428)
					(10, 429)
					(10, 430)
					(10, 431)
					(10, 432)
					(10, 433)
					(10, 434)
					(10, 435)
					(10, 436)
					(10, 437)
					(10, 438)
					(10, 439)
					(10, 440)
					(10, 441)
					(10, 442)
					(10, 443)
					(11, 418)
					(11, 419)
					(11, 420)
					(11, 421)
					(11, 422)
					(11, 423)
					(11, 424)
					(11, 425)
					(11, 426)
					(11, 427)
					(11, 428)
					(11, 429)
					(11, 430)
					(11, 431)
					(11, 432)
					(11, 433)
					(11, 434)
					(11, 435)
					(11, 436)
					(11, 437)
					(11, 438)
					(11, 439)
					(11, 440)
					(11, 441)
					(12, 417)
					(12, 418)
					(12, 419)
					(12, 420)
					(12, 421)
					(12, 422)
					(12, 423)
					(12, 424)
					(12, 425)
					(12, 426)
					(12, 427)
					(12, 428)
					(12, 429)
					(12, 430)
					(12, 431)
					(12, 432)
					(12, 433)
					(12, 434)
					(12, 435)
					(12, 436)
					(12, 437)
					(13, 415)
					(13, 416)
					(13, 417)
					(13, 418)
					(13, 419)
					(13, 420)
					(13, 421)
					(13, 422)
					(13, 423)
					(13, 424)
					(13, 425)
					(13, 426)
					(13, 427)
					(13, 428)
					(13, 429)
					(13, 430)
					(13, 431)
					(13, 432)
					(13, 433)
					(13, 434)
					(14, 412)
					(14, 413)
					(14, 414)
					(14, 415)
					(14, 416)
					(14, 417)
					(14, 418)
					(14, 419)
					(14, 420)
					(14, 421)
					(14, 422)
					(14, 423)
					(14, 424)
					(14, 425)
					(14, 426)
					(14, 427)
					(14, 428)
					(14, 429)
					(14, 430)
					(14, 431)
					(14, 432)
					(15, 411)
					(15, 412)
					(15, 413)
					(15, 414)
					(15, 415)
					(15, 416)
					(15, 417)
					(15, 418)
					(15, 419)
					(15, 420)
					(15, 421)
					(15, 422)
					(15, 423)
					(15, 424)
					(15, 425)
					(15, 426)
					(15, 427)
					(15, 428)
					(15, 429)
					(16, 409)
					(16, 410)
					(16, 411)
					(16, 412)
					(16, 413)
					(16, 414)
					(16, 415)
					(16, 416)
					(16, 417)
					(16, 418)
					(16, 419)
					(16, 420)
					(16, 421)
					(16, 422)
					(16, 423)
					(16, 424)
					(16, 425)
					(16, 426)
					(16, 427)
					(17, 408)
					(17, 409)
					(17, 410)
					(17, 411)
					(17, 412)
					(17, 413)
					(17, 414)
					(17, 415)
					(17, 416)
					(17, 417)
					(17, 418)
					(17, 419)
					(17, 420)
					(17, 421)
					(17, 422)
					(17, 423)
					(17, 424)
					(17, 425)
					(18, 407)
					(18, 408)
					(18, 409)
					(18, 410)
					(18, 411)
					(18, 412)
					(18, 413)
					(18, 414)
					(18, 415)
					(18, 416)
					(18, 417)
					(18, 418)
					(18, 419)
					(18, 420)
					(18, 421)
					(18, 422)
					(18, 423)
					(19, 406)
					(19, 407)
					(19, 408)
					(19, 409)
					(19, 410)
					(19, 411)
					(19, 412)
					(19, 413)
					(19, 414)
					(19, 415)
					(19, 416)
					(19, 417)
					(19, 418)
					(19, 419)
					(19, 420)
					(19, 421)
					(20, 405)
					(20, 406)
					(20, 407)
					(20, 408)
					(20, 409)
					(20, 410)
					(20, 411)
					(20, 412)
					(20, 413)
					(20, 414)
					(20, 415)
					(20, 416)
					(20, 417)
					(20, 418)
					(20, 419)
					(21, 403)
					(21, 404)
					(21, 405)
					(21, 406)
					(21, 407)
					(21, 408)
					(21, 409)
					(21, 410)
					(21, 411)
					(21, 412)
					(21, 413)
					(21, 414)
					(21, 415)
					(21, 416)
					(21, 417)
					(22, 402)
					(22, 403)
					(22, 404)
					(22, 405)
					(22, 406)
					(22, 407)
					(22, 408)
					(22, 409)
					(22, 410)
					(22, 411)
					(22, 412)
					(22, 413)
					(22, 414)
					(22, 415)
					(23, 400)
					(23, 401)
					(23, 402)
					(23, 403)
					(23, 404)
					(23, 405)
					(23, 406)
					(23, 407)
					(23, 408)
					(23, 409)
					(23, 410)
					(23, 411)
					(23, 412)
					(23, 413)
					(24, 399)
					(24, 400)
					(24, 401)
					(24, 402)
					(24, 403)
					(24, 404)
					(24, 405)
					(24, 406)
					(24, 407)
					(24, 408)
					(24, 409)
					(24, 410)
					(24, 411)
					(24, 412)
					(25, 398)
					(25, 399)
					(25, 400)
					(25, 401)
					(25, 402)
					(25, 403)
					(25, 404)
					(25, 405)
					(25, 406)
					(25, 407)
					(25, 408)
					(25, 409)
					(25, 410)
					(26, 397)
					(26, 398)
					(26, 399)
					(26, 400)
					(26, 401)
					(26, 402)
					(26, 403)
					(26, 404)
					(26, 405)
					(26, 406)
					(26, 407)
					(26, 408)
					(27, 396)
					(27, 397)
					(27, 398)
					(27, 399)
					(27, 400)
					(27, 401)
					(27, 402)
					(27, 403)
					(27, 404)
					(27, 405)
					(27, 406)
					(27, 407)
					(28, 395)
					(28, 396)
					(28, 397)
					(28, 398)
					(28, 399)
					(28, 400)
					(28, 401)
					(28, 402)
					(28, 403)
					(28, 404)
					(28, 405)
					(29, 394)
					(29, 395)
					(29, 396)
					(29, 397)
					(29, 398)
					(29, 399)
					(29, 400)
					(29, 401)
					(29, 402)
					(29, 403)
					(29, 404)
					(30, 393)
					(30, 394)
					(30, 395)
					(30, 396)
					(30, 397)
					(30, 398)
					(30, 399)
					(30, 400)
					(30, 401)
					(30, 402)
					(30, 403)
					(31, 392)
					(31, 393)
					(31, 394)
					(31, 395)
					(31, 396)
					(31, 397)
					(31, 398)
					(31, 399)
					(31, 400)
					(31, 401)
					(32, 390)
					(32, 391)
					(32, 392)
					(32, 393)
					(32, 394)
					(32, 395)
					(32, 396)
					(32, 397)
					(32, 398)
					(32, 399)
					(33, 389)
					(33, 390)
					(33, 391)
					(33, 392)
					(33, 393)
					(33, 394)
					(33, 395)
					(33, 396)
					(33, 397)
					(34, 388)
					(34, 389)
					(34, 390)
					(34, 391)
					(34, 392)
					(34, 393)
					(34, 394)
					(34, 395)
					(34, 396)
					(35, 387)
					(35, 388)
					(35, 389)
					(35, 390)
					(35, 391)
					(35, 392)
					(35, 393)
					(35, 394)
					(36, 386)
					(36, 387)
					(36, 388)
					(36, 389)
					(36, 390)
					(36, 391)
					(36, 392)
					(36, 393)
					(37, 385)
					(37, 386)
					(37, 387)
					(37, 388)
					(37, 389)
					(37, 390)
					(37, 391)
					(37, 392)
					(38, 384)
					(38, 385)
					(38, 386)
					(38, 387)
					(38, 388)
					(38, 389)
					(38, 390)
					(38, 391)
					(39, 383)
					(39, 384)
					(39, 385)
					(39, 386)
					(39, 387)
					(39, 388)
					(39, 389)
					(40, 382)
					(40, 383)
					(40, 384)
					(40, 385)
					(40, 386)
					(40, 387)
					(40, 388)
					(41, 381)
					(41, 382)
					(41, 383)
					(41, 384)
					(41, 385)
					(41, 386)
					(41, 387)
					(42, 380)
					(42, 381)
					(42, 382)
					(42, 383)
					(42, 384)
					(42, 385)
					(43, 379)
					(43, 380)
					(43, 381)
					(43, 382)
					(43, 383)
					(43, 384)
					(44, 378)
					(44, 379)
					(44, 380)
					(44, 381)
					(44, 382)
					(44, 383)
					(45, 377)
					(45, 378)
					(45, 379)
					(45, 380)
					(45, 381)
					(46, 376)
					(46, 377)
					(46, 378)
					(46, 379)
					(46, 380)
					(47, 375)
					(47, 376)
					(47, 377)
					(47, 378)
					(48, 374)
					(48, 375)
					(48, 376)
					(48, 377)
					(49, 373)
					(49, 374)
					(49, 375)
					(49, 376)
					(50, 372)
					(50, 373)
					(50, 374)
					(50, 375)
					(51, 371)
					(51, 372)
					(51, 373)
					(51, 374)
					(52, 370)
					(52, 371)
					(52, 372)
					(52, 373)
					(53, 369)
					(53, 370)
					(53, 371)
					(53, 372)
					(54, 368)
					(54, 369)
					(54, 370)
					(54, 371)
					(55, 367)
					(55, 368)
					(55, 369)
					(55, 370)
					(56, 366)
					(56, 367)
					(56, 368)
					(56, 369)
					(57, 365)
					(57, 366)
					(57, 367)
					(57, 368)
					(58, 364)
					(58, 365)
					(58, 366)
					(59, 363)
					(59, 364)
					(59, 365)
					(60, 362)
					(60, 363)
					(60, 364)
					(61, 361)
					(61, 362)
					(62, 360)
					(62, 361)
					(63, 359)
					(63, 360)
					(64, 358)
					(64, 359)
					(65, 357)
					(65, 358)
					(66, 356)
					(66, 357)
					(67, 355)
					(67, 356)
					(68, 354)
					(68, 355)
					(69, 353)
					(69, 354)
					(70, 352)
					(70, 353)
					(71, 351)
					(71, 352)
					(72, 350)
					(72, 351)
					(73, 349)
					(73, 350)
					(74, 348)
					(74, 349)
					(75, 347)
					(75, 348)
					(76, 346)
					(77, 345)
					(78, 344)
					(79, 343)
					(80, 342)
					(81, 341)
					(82, 340)
					(83, 339)
					(84, 338)
					(85, 337)
					(86, 336)
					(87, 335)
					(88, 334)
					(89, 333)
					(90, 332)
					(91, 331)
					(92, 330)
					(93, 329)
					(94, 328)
					(95, 327)
					(96, 326)
					(97, 325)
					(98, 324)
					(99, 323)
					(100, 322)
					(101, 321)
					(102, 320)
					(103, 319)
					(104, 318)
					(105, 317)
					(106, 316)
					(107, 315)
					(108, 314)
					(109, 313)
					(110, 312)
					(111, 311)
					(112, 310)
					(113, 309)
					(114, 308)
					(115, 307)
					(116, 306)
					(117, 305)
					(118, 304)
					(119, 303)
					(120, 302)
					(121, 301)
					(122, 300)
					(123, 299)
					(124, 298)
					(125, 297)
					(126, 296)
					(127, 295)
					(128, 294)
					(129, 293)
					(130, 292)
					(131, 291)
					(132, 290)
					(133, 289)
					(134, 288)
					(135, 287)
					(136, 286)
					(137, 285)
					(138, 284)
					(139, 283)
					(140, 282)
					(141, 281)
					(142, 280)
					(143, 279)
					(144, 278)
					(145, 277)
					(146, 276)
					(147, 275)
					(148, 274)
					(149, 273)
					(150, 272)
					(151, 271)
					(152, 270)
					(153, 269)
					(154, 268)
					(155, 267)
					(156, 266)
					(157, 265)
					(158, 264)
					(159, 263)
					(160, 262)
					(161, 261)
					(162, 260)
					(163, 259)
					(164, 258)
					(165, 257)
					(166, 256)
					(167, 255)
					(168, 254)
					(169, 253)
					(170, 252)
					(171, 251)
					(172, 250)
					(173, 249)
					(174, 248)
					(175, 247)
					(176, 246)
					(177, 245)
					(178, 244)
					(179, 243)
					(180, 242)
					(181, 241)
					(182, 240)
					(183, 239)
					(184, 238)
					(185, 237)
					(186, 236)
					(187, 235)
					(188, 234)
					(189, 233)
					(190, 232)
					(191, 231)
					(192, 230)
					(193, 229)
					(194, 228)
					(195, 227)
					(196, 226)
					(197, 225)
					(198, 224)
					(199, 223)
					(200, 222)
					(201, 221)
					(202, 220)
					(203, 219)
					(204, 218)
					(205, 217)
					(206, 216)
					(207, 215)
					(208, 214)
					(209, 213)
					(210, 212)
					(211, 211)
					(212, 210)
					(213, 209)
					(214, 208)
					(215, 207)
					(216, 206)
					(217, 205)
					(218, 204)
					(219, 203)
					(220, 202)
					(221, 201)
					(222, 200)
					(223, 199)
					(224, 198)
					(225, 197)
					(226, 196)
					(227, 195)
					(228, 194)
					(229, 193)
					(230, 192)
					(231, 191)
					(232, 190)
					(233, 189)
					(234, 188)
					(235, 187)
					(236, 186)
					(237, 185)
					(238, 184)
					(239, 183)
					(240, 182)
					(241, 181)
					(242, 180)
					(243, 179)
					(244, 178)
					(245, 177)
					(246, 176)
					(247, 175)
					(248, 174)
					(249, 173)
					(250, 172)
					(251, 171)
					(252, 170)
					(253, 169)
					(254, 168)
					(255, 167)
					(256, 166)
					(257, 165)
					(258, 164)
					(259, 163)
					(260, 162)
					(261, 161)
					(262, 160)
					(263, 159)
					(264, 158)
					(265, 157)
					(266, 156)
					(267, 155)
					(268, 154)
					(269, 153)
					(270, 152)
					(271, 151)
					(272, 150)
					(273, 149)
					(274, 148)
					(275, 147)
					(276, 146)
					(277, 145)
					(278, 144)
					(279, 143)
					(280, 142)
					(281, 141)
					(282, 140)
					(283, 139)
					(284, 138)
					(285, 137)
					(286, 136)
					(287, 135)
					(288, 134)
					(289, 133)
					(290, 132)
					(291, 131)
					(292, 130)
					(293, 129)
					(294, 128)
					(295, 127)
					(296, 126)
					(297, 125)
					(298, 124)
					(299, 123)
					(300, 122)
					(301, 121)
					(302, 120)
					(303, 119)
					(304, 118)
					(305, 117)
					(306, 116)
					(307, 115)
					(308, 114)
					(309, 113)
					(310, 112)
					(311, 111)
					(312, 110)
					(313, 109)
					(314, 108)
					(315, 107)
					(316, 106)
					(317, 105)
					(318, 104)
					(319, 103)
					(320, 102)
					(321, 101)
					(322, 100)
					(323, 99)
					(324, 98)
					(325, 97)
					(326, 96)
					(327, 95)
					(328, 94)
					(329, 93)
					(330, 92)
					(331, 91)
					(332, 90)
					(333, 89)
					(334, 88)
					(335, 87)
					(336, 86)
					(337, 85)
					(338, 84)
					(339, 83)
					(340, 82)
					(341, 81)
					(342, 80)
					(343, 79)
					(344, 78)
					(345, 77)
					(346, 76)
					(347, 75)
					(348, 74)
					(349, 73)
					(350, 72)
					(351, 71)
					(352, 70)
					(353, 69)
					(354, 68)
					(355, 67)
					(356, 66)
					(357, 65)
					(358, 64)
					(359, 63)
					(360, 62)
					(361, 61)
					(362, 60)
					(363, 59)
					(364, 58)
					(365, 57)
					(366, 56)
					(367, 55)
					(368, 54)
					(369, 53)
					(370, 52)
					(371, 51)
					(372, 50)
					(373, 49)
					(374, 48)
					(375, 47)
					(376, 46)
					(377, 45)
					(378, 44)
					(379, 43)
					(380, 42)
					(381, 41)
					(382, 40)
					(383, 39)
					(384, 38)
					(385, 37)
					(386, 36)
					(387, 35)
					(388, 34)
					(389, 33)
					(390, 32)
					(391, 31)
					(392, 30)
					(393, 29)
					(394, 28)
					(395, 27)
					(396, 26)
					(397, 25)
					(398, 24)
					(399, 23)
					(400, 22)
					(401, 21)
					(402, 20)
					(403, 19)
					(404, 18)
					(405, 17)
					(406, 16)
					(407, 15)
					(408, 14)
					(409, 13)
					(410, 12)
					(411, 11)
					(412, 10)
					(413, 9)
					(414, 8)
					(415, 7)
					(416, 6)
					(417, 5)
					(418, 4)
					(419, 3)
					(420, 2)
					(421, 1)
					(422, 0)
					(423, 0)
					(424, 0)
					(425, 0)
					(426, 0)
					(427, 0)
					(428, 0)
					(429, 0)
					(430, 0)
					(431, 0)
					(432, 0)
					(433, 0)
					(434, 0)
					(435, 0)
					(436, 0)
					(437, 0)
					(438, 0)
					(439, 0)
					(440, 0)
					(441, 0)
					(442, 0)
					(443, 0)
					(444, 0)
					(445, 0)
					(446, 0)
					(447, 0)
					(448, 0)
					(449, 0)
					(450, 0)
					(451, 0)
					(452, 0)
					(453, 0)
					(454, 0)
					(455, 0)
					(456, 0)
					(457, 0)
					(458, 0)
					(459, 0)
					(460, 0)
					(461, 0)
					(462, 0)
					(463, 0)
					(464, 0)
					(465, 0)
					(466, 0)
					(467, 0)
					(468, 0)
					(469, 0)
					(470, 0)
					(471, 0)
					(472, 0)
					(473, 0)
					(474, 0)
					(475, 0)
					(476, 0)
					(477, 0)
					(478, 0)
					(479, 0)
					(480, 0)
					(481, 0)
					(482, 0)
					(483, 0)
					(484, 0)
					(485, 0)
					(486, 0)
					(487, 0)
					(488, 0)
					(489, 0)
					(490, 0)
					(491, 0)
					(492, 0)
					(493, 0)
					(494, 0)
					(495, 0)
					(496, 0)
				};
				
				\addplot[blue, mark options={
					fill=blue,scale=0.08},mark=o] coordinates {
					(493, 1)
					(491, 2)
					(488, 3)
					(486, 4)
					(483, 5)
					(481, 6)
					(479, 7)
					(477, 8)
					(475, 9)
					(472, 10)
					(470, 11)
					(468, 12)
					(466, 13)
					(464, 14)
					(462, 15)
					(460, 16)
					(458, 17)
					(456, 18)
					(454, 19)
					(453, 20)
					(451, 21)
					(449, 22)
					(447, 23)
					(445, 24)
					(443, 25)
					(441, 26)
					(439, 27)
					(438, 28)
					(436, 29)
					(434, 30)
					(432, 31)
					(431, 32)
					(429, 33)
					(427, 34)
					(425, 35)
					(424, 36)
					(422, 37)
					(420, 38)
					(418, 39)
					(417, 40)
					(415, 41)
					(413, 42)
					(412, 43)
					(410, 44)
					(408, 45)
					(407, 46)
					(405, 47)
					(403, 48)
					(402, 49)
					(400, 50)
					(399, 51)
					(397, 52)
					(395, 53)
					(394, 54)
					(392, 55)
					(391, 56)
					(389, 57)
					(387, 58)
					(386, 59)
					(384, 60)
					(383, 61)
					(381, 62)
					(380, 63)
					(378, 64)
					(377, 65)
					(375, 66)
					(373, 67)
					(372, 68)
					(370, 69)
					(369, 70)
					(367, 71)
					(366, 72)
					(364, 73)
					(363, 74)
					(361, 75)
					(360, 76)
					(358, 77)
					(357, 78)
					(355, 79)
					(354, 80)
					(353, 81)
					(351, 82)
					(350, 83)
					(348, 84)
					(347, 85)
					(345, 86)
					(344, 87)
					(342, 88)
					(341, 89)
					(340, 90)
					(338, 91)
					(337, 92)
					(335, 93)
					(334, 94)
					(332, 95)
					(331, 96)
					(330, 97)
					(328, 98)
					(327, 99)
					(325, 100)
					(324, 101)
					(323, 102)
					(321, 103)
					(320, 104)
					(319, 105)
					(317, 106)
					(316, 107)
					(314, 108)
					(313, 109)
					(312, 110)
					(310, 111)
					(309, 112)
					(308, 113)
					(306, 114)
					(305, 115)
					(304, 116)
					(302, 117)
					(301, 118)
					(300, 119)
					(298, 120)
					(297, 121)
					(296, 122)
					(294, 123)
					(293, 124)
					(292, 125)
					(291, 126)
					(289, 127)
					(288, 128)
					(287, 129)
					(285, 130)
					(284, 131)
					(283, 132)
					(282, 133)
					(280, 134)
					(279, 135)
					(278, 136)
					(276, 137)
					(275, 138)
					(274, 139)
					(273, 140)
					(271, 141)
					(270, 142)
					(269, 143)
					(268, 144)
					(266, 145)
					(265, 146)
					(264, 147)
					(263, 148)
					(261, 149)
					(260, 150)
					(259, 151)
					(258, 152)
					(256, 153)
					(255, 154)
					(254, 155)
					(253, 156)
					(252, 157)
					(250, 158)
					(249, 159)
					(248, 160)
					(247, 161)
					(246, 162)
					(244, 163)
					(243, 164)
					(242, 165)
					(241, 166)
					(240, 167)
					(238, 168)
					(237, 169)
					(236, 170)
					(235, 171)
					(234, 172)
					(233, 173)
					(231, 174)
					(230, 175)
					(229, 176)
					(228, 177)
					(227, 178)
					(226, 179)
					(224, 180)
					(223, 181)
					(222, 182)
					(221, 183)
					(220, 184)
					(219, 185)
					(218, 186)
					(216, 187)
					(215, 188)
					(214, 189)
					(213, 190)
					(212, 191)
					(211, 192)
					(210, 193)
					(208, 194)
					(207, 195)
					(206, 196)
					(205, 197)
					(204, 198)
					(203, 199)
					(202, 200)
					(201, 201)
					(200, 202)
					(199, 203)
					(197, 204)
					(196, 205)
					(195, 206)
					(194, 207)
					(193, 208)
					(192, 209)
					(191, 210)
					(190, 211)
					(189, 212)
					(188, 213)
					(187, 214)
					(186, 215)
					(184, 216)
					(183, 217)
					(182, 218)
					(181, 219)
					(180, 220)
					(179, 221)
					(178, 222)
					(177, 223)
					(176, 224)
					(175, 225)
					(174, 226)
					(173, 227)
					(172, 228)
					(171, 229)
					(170, 230)
					(169, 231)
					(168, 232)
					(167, 233)
					(166, 234)
					(165, 235)
					(164, 236)
					(163, 237)
					(162, 238)
					(161, 239)
					(160, 240)
					(159, 241)
					(158, 242)
					(157, 243)
					(156, 244)
					(155, 245)
					(154, 246)
					(153, 247)
					(152, 248)
					(151, 249)
					(150, 250)
					(149, 251)
					(148, 252)
					(147, 253)
					(146, 254)
					(145, 255)
					(144, 256)
					(143, 257)
					(142, 258)
					(141, 259)
					(140, 260)
					(139, 261)
					(138, 262)
					(137, 263)
					(136, 264)
					(135, 265)
					(134, 266)
					(133, 267)
					(132, 268)
					(131, 269)
					(130, 270)
					(129, 271)
					(128, 272)
					(128, 273)
					(127, 274)
					(126, 275)
					(125, 276)
					(124, 277)
					(123, 278)
					(122, 279)
					(121, 280)
					(120, 281)
					(119, 282)
					(118, 283)
					(117, 284)
					(117, 285)
					(116, 286)
					(115, 287)
					(114, 288)
					(113, 289)
					(112, 290)
					(111, 291)
					(110, 292)
					(109, 293)
					(108, 294)
					(108, 295)
					(107, 296)
					(106, 297)
					(105, 298)
					(104, 299)
					(103, 300)
					(102, 301)
					(101, 302)
					(101, 303)
					(100, 304)
					(99, 305)
					(98, 306)
					(97, 307)
					(96, 308)
					(96, 309)
					(95, 310)
					(94, 311)
					(93, 312)
					(92, 313)
					(91, 314)
					(90, 315)
					(90, 316)
					(89, 317)
					(88, 318)
					(87, 319)
					(86, 320)
					(86, 321)
					(85, 322)
					(84, 323)
					(83, 324)
					(82, 325)
					(81, 326)
					(81, 327)
					(80, 328)
					(79, 329)
					(78, 330)
					(78, 331)
					(77, 332)
					(76, 333)
					(75, 334)
					(74, 335)
					(74, 336)
					(73, 337)
					(72, 338)
					(71, 339)
					(71, 340)
					(70, 341)
					(69, 342)
					(68, 343)
					(67, 344)
					(67, 345)
					(66, 346)
					(65, 347)
					(64, 348)
					(64, 349)
					(63, 350)
					(62, 351)
					(62, 352)
					(61, 353)
					(60, 354)
					(59, 355)
					(59, 356)
					(58, 357)
					(57, 358)
					(56, 359)
					(56, 360)
					(55, 361)
					(54, 362)
					(54, 363)
					(53, 364)
					(52, 365)
					(52, 366)
					(51, 367)
					(50, 368)
					(49, 369)
					(49, 370)
					(48, 371)
					(47, 372)
					(47, 373)
					(46, 374)
					(45, 375)
					(45, 376)
					(44, 377)
					(43, 378)
					(43, 379)
					(42, 380)
					(41, 381)
					(41, 382)
					(40, 383)
					(40, 384)
					(39, 385)
					(38, 386)
					(38, 387)
					(37, 388)
					(36, 389)
					(36, 390)
					(35, 391)
					(35, 392)
					(34, 393)
					(33, 394)
					(33, 395)
					(32, 396)
					(32, 397)
					(31, 398)
					(30, 399)
					(30, 400)
					(29, 401)
					(29, 402)
					(28, 403)
					(28, 404)
					(27, 405)
					(26, 406)
					(26, 407)
					(25, 408)
					(25, 409)
					(24, 410)
					(24, 411)
					(23, 412)
					(23, 413)
					(22, 414)
					(22, 415)
					(21, 416)
					(21, 417)
					(20, 418)
					(20, 419)
					(19, 420)
					(19, 421)
					(18, 422)
					(18, 423)
					(17, 424)
					(17, 425)
					(16, 426)
					(16, 427)
					(15, 428)
					(15, 429)
					(14, 430)
					(14, 431)
					(13, 432)
					(13, 433)
					(13, 434)
					(12, 435)
					(12, 436)
					(11, 437)
					(11, 438)
					(10, 439)
					(10, 440)
					(10, 441)
					(9, 442)
					(9, 443)
					(9, 444)
					(8, 445)
					(8, 446)
					(7, 447)
					(7, 448)
					(7, 449)
					(6, 450)
					(6, 451)
					(6, 452)
					(5, 453)
					(5, 454)
					(5, 455)
					(5, 456)
					(4, 457)
					(4, 458)
					(4, 459)
					(3, 460)
					(3, 461)
					(3, 462)
					(3, 463)
					(3, 464)
					(2, 465)
					(2, 466)
					(2, 467)
					(2, 468)
					(2, 469)
					(1, 470)
					(1, 471)
					(1, 472)
					(1, 473)
					(1, 474)
					(1, 475)
					(1, 476)
					(1, 477)
					(1, 478)
					(1, 479)
					(1, 480)
				};

				\end{axis}
				\end{tikzpicture}
				\protect\caption{Bound for $\mathbb{F}_{q^c} = \mathbb{F}_{32} $ }
					\label{fig:GV_bound}
			\end{figure}
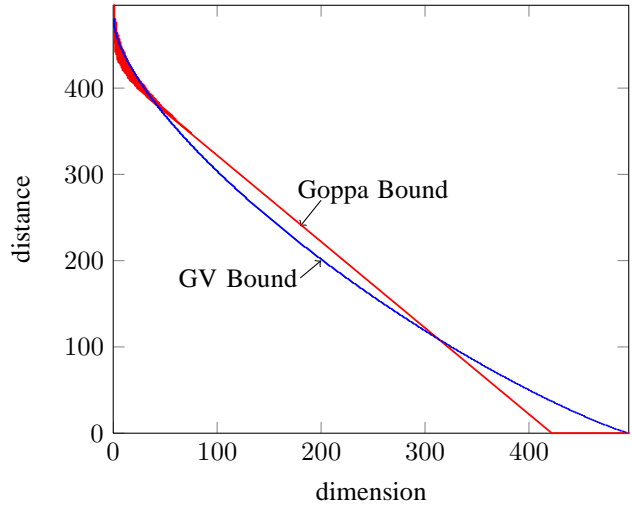

\end{example}

\ifCLASSOPTIONcaptionsoff
  \newpage
\fi

% trigger a \newpage just before the given reference
% number - used to balance the columns on the last page
% adjust value as needed - may need to be readjusted if
% the document is modified later
%\IEEEtriggeratref{8}
% The "triggered" command can be changed if desired:
%\IEEEtriggercmd{\enlargethispage{-5in}}

% references section

% can use a bibliography generated by BibTeX as a .bbl file
% BibTeX documentation can be easily obtained at:
% http://www.ctan.org/tex-archive/biblio/bibtex/contrib/doc/
% The IEEEtran BibTeX style support page is at:
% http://www.michaelshell.org/tex/ieeetran/bibtex/
%\bibliographystyle{IEEEtran}
% argument is your BibTeX string definitions and bibliography database(s)
%\bibliography{IEEEabrv,../bib/paper}
%
% <OR> manually copy in the resultant .bbl file
% set second argument of \begin to the number of references
% (used to reserve space for the reference number labels box)
\bibliographystyle{IEEEtran}
% argument is your BibTeX string definitions and bibliography database(s)
\bibliography{paper}

% biography section
% 
% If you have an EPS/PDF photo (graphicx package needed) extra braces are
% needed around the contents of the optional argument to biography to prevent
% the LaTeX parser from getting confused when it sees the complicated
% \includegraphics command within an optional argument. (You could create
% your own custom macro containing the \includegraphics command to make things
% simpler here.)
%\begin{IEEEbiography}[{\includegraphics[width=1in,height=1.25in, clip,keepaspectratio]{mshell}}]{Michael Shell}
% or if you just want to reserve a space for a photo:

% if you will not have a photo at all:

% insert where needed to balance the two columns on the last page with
% biographies
%\newpage

% You can push biographies down or up by placing
% a \vfill before or after them. The appropriate
% use of \vfill depends on what kind of text is
% on the last page and whether or not the columns
% are being equalized.

%\vfill

% Can be used to pull up biographies so that the bottom of the last one
% is flush with the other column.
%\enlargethispage{-5in}

% that's all folks
\end{document}